\documentclass[preprint,twocolumn]{elsarticle}
\pdfoutput=1

\newcommand{\USETIKZ}{}

\journal{Information Processing Letters}

\usepackage[utf8]{inputenc}

\usepackage[margin=2.25cm]{geometry}

\usepackage{amsmath}
\usepackage{amssymb}
\usepackage{amsthm}

\usepackage{thmtools}

\usepackage{enumerate}

\usepackage{algorithmic}
\usepackage[linesnumbered,titlenumbered,ruled,vlined,resetcount,algosection]{algorithm2e}

\usepackage[usenames,dvipsnames]{xcolor}
\usepackage{longtable}
\usepackage{booktabs}
\usepackage{graphicx}
\usepackage{multirow}
\usepackage{float}
\usepackage{color}

\usepackage{xparse} 
\usepackage{xifthen} 
\usepackage{ifmtarg} 

\usepackage[labelformat=simple]{subcaption}

\usepackage{todonotes}

\usepackage[hidelinks]{hyperref}
\hypersetup{
	colorlinks,
	linkcolor={black!100!black},
	citecolor={blue!100!black},
	urlcolor={blue!50!black}
}

\usepackage{pdflscape}

\usepackage{pdfcomment}
\usepackage[normalem]{ulem}

\usepackage{etoolbox}

\makeatletter
\def\IfEmptyTF#1{%
  \if\relax\detokenize{#1}\relax
    \expandafter\@firstoftwo
  \else
    \expandafter\@secondoftwo
  \fi}
\makeatother

\ifdefined\USETIKZ
\usepackage{tikz}
\fi

\newcommand \red[1] 	{\textcolor{red}{#1}}

\newcommand \suppref[1] 	{{\em\red{#1}}}

\newcommand \projective {\mathrm{pr}}
\newcommand \planar {\mathrm{pl}}

\NewDocumentCommand \bigO { m } {O(#1)}

\newtheoremstyle{theoremstyle}
    {15pt} 
    {} 
    {\itshape} 
    {} 
    {\bfseries} 
    {.} 
    {.5em} 
    {} 

\theoremstyle{theoremstyle}
\newtheorem{lemma}{Lemma}[section]
\newtheorem{theorem}[lemma]{Theorem}
\newtheorem{proposition}[theorem]{Proposition}
\newtheorem{corollary}[proposition]{Corollary}
\newtheorem{conjecture}[corollary]{Conjecture}
\newtheorem{property}[conjecture]{Property}

\NewDocumentCommand \Root { } {r}

\NewDocumentCommand \Ftree { O{T} }{ #1 }
\NewDocumentCommand \Rtree { O{\Root} O{\Ftree} }{ {#2}^{#1} }
\NewDocumentCommand \SubRtree { m O{\Root} O{\Ftree} }{ #3_{#1}^{#2} }

\NewDocumentCommand \RootedTrees { O{n} }{ \mathcal{R}_{#1} }
\NewDocumentCommand \FreeTrees { O{n} }{ \mathcal{F}_{#1} }

\NewDocumentCommand \BalancedBistar { O{\Ftree} }{ {#1}^{\mathrm{bb}} }
\NewDocumentCommand \QuasiStar { O{\Ftree} }{ {#1}^{\mathrm{qstar}} }
\NewDocumentCommand \Star { O{\Ftree} }{ {#1}^{\mathrm{star}} }
\NewDocumentCommand \StarHub { O{\Ftree} }{ {#1}^{\mathrm{hub,star}} }
\NewDocumentCommand \Linear { O{\Ftree} }{ {#1}^{\mathrm{linear}} }

\NewDocumentCommand \maxroots { O{\Ftree} } {V_*(#1)}
\NewDocumentCommand \maxrootsk { O{\Ftree} } {V_1(#1)}

\NewDocumentCommand \Nvert { m m }				{ s_{#1}(#2) }
\NewDocumentCommand \NvertT { m m O{\Ftree} }	{ s_{\Rtree[#1][#3]}(#2) }
\NewDocumentCommand \numleaves { O{\Ftree} }	{ L(#1) }

\NewDocumentCommand \largestNvert { m m m } { s_{#1}(#2,#3) }
\NewDocumentCommand \largestNvertT { m m m O{\Ftree} } { s_{\Rtree[#1][#4]}(#2,#3) }
\NewDocumentCommand \singlelargestNvert { m m } { s_{#1}(#1,#2) }
\NewDocumentCommand \indexofT { m m O{\Ftree} } { \sigma_{\Rtree[#1][#3]}(#2) }

\NewDocumentCommand \degreesymbol { } { d }
\NewDocumentCommand \degreeT { m O{\Ftree} } { \degreesymbol_{#2}(#1) }
\NewDocumentCommand \degree { m } { \degreesymbol(#1) }
\NewDocumentCommand \outdegreeT { m O{\Root} O{\Ftree} } { \degreesymbol_{\Rtree[#2][#3]}(#1) }
\NewDocumentCommand \outdegree { m O{\Root} } { \degreesymbol_{#2}(#1) }

\NewDocumentCommand \routdegreeT { O{\Root} } { \outdegreeT{#1}[#1] }
\NewDocumentCommand \routdegree { O{\Root} } { \outdegree{#1}[#1] }
\NewDocumentCommand \routdegreeTshort { O{\Root} } { d_{#1} }

\NewDocumentCommand \parent { m O{\Root} } { p(#2,#1) }

\NewDocumentCommand \neighbours {} { \Gamma }
\NewDocumentCommand \neighsT { m O{\Ftree} } { \neighbours_{#2}(#1) }

\NewDocumentCommand \outneighsT { m O{\Root} O{\Ftree} } { \neighbours_{\Rtree[#2][#3]}(#1) }
\NewDocumentCommand \outneighs { m O{\Root} } { \neighbours_{#2}(#1) }
\NewDocumentCommand \RoutneighsT { O{\Root} O{\Ftree} } { \neighbours_{\Rtree[#1][#2]}(#1) }

\NewDocumentCommand \arr { } {\pi}

\NewDocumentCommand \dl { m O{\arr} } { d_{#2}(#1) }

\NewDocumentCommand \D { m O{\arr} } { D_{#2}(#1) }

\NewDocumentCommand \gDmin {m m} {m_{#1}\left[ #2 \right]}
\NewDocumentCommand \DminProj {O{\Rtree}} {\gDmin{\projective}{#1}}
\NewDocumentCommand \DminPlan {O{\Ftree}} {\gDmin{\planar}{#1}}
\NewDocumentCommand \Dmin {O{\Ftree}} {\gDmin{}{#1}}
\NewDocumentCommand \gDMax {m m} {M_{#1}\left[ #2 \right]}
\NewDocumentCommand \DMaxProj {O{\Rtree}} {\gDMax{\projective}{#1}}
\NewDocumentCommand \DMaxPlan {O{\Ftree}} {\gDMax{\planar}{#1}}
\NewDocumentCommand \DMax {O{\Ftree}} {\gDMax{}{#1}}

\NewDocumentCommand \subtreesizelist { O{} } {\mathcal{L}_{#1}}
\newcommand \magadjlist {\mathcal{M}}

\ifdefined\USETIKZ

\usetikzlibrary{arrows.meta,patterns}
\usetikzlibrary{ipe} 

\tikzstyle{ipe stylesheet} = [
  ipe import,
  even odd rule,
  line join=round,
  line cap=butt,
  ipe pen normal/.style={line width=0.4},
  ipe pen heavier/.style={line width=0.8},
  ipe pen fat/.style={line width=1.2},
  ipe pen ultrafat/.style={line width=2},
  ipe pen normal,
  ipe mark normal/.style={ipe mark scale=3},
  ipe mark large/.style={ipe mark scale=5},
  ipe mark small/.style={ipe mark scale=2},
  ipe mark tiny/.style={ipe mark scale=1.1},
  ipe mark normal,
  /pgf/arrow keys/.cd,
  ipe arrow normal/.style={scale=7},
  ipe arrow large/.style={scale=10},
  ipe arrow small/.style={scale=5},
  ipe arrow tiny/.style={scale=3},
  ipe arrow normal,
  /tikz/.cd,
  ipe arrows, 
  <->/.tip = ipe normal,
  ipe dash normal/.style={dash pattern=},
  ipe dash dotted/.style={dash pattern=on 1bp off 3bp},
  ipe dash dashed/.style={dash pattern=on 4bp off 4bp},
  ipe dash dash dotted/.style={dash pattern=on 4bp off 2bp on 1bp off 2bp},
  ipe dash dash dot dotted/.style={dash pattern=on 4bp off 2bp on 1bp off 2bp on 1bp off 2bp},
  ipe dash normal,
  ipe node/.append style={font=\normalsize},
  ipe stretch normal/.style={ipe node stretch=1},
  ipe stretch normal,
  ipe opacity 10/.style={opacity=0.1},
  ipe opacity 30/.style={opacity=0.3},
  ipe opacity 50/.style={opacity=0.5},
  ipe opacity 75/.style={opacity=0.75},
  ipe opacity opaque/.style={opacity=1},
  ipe opacity opaque,
]

\definecolor{red}{rgb}{1,0,0}
\definecolor{blue}{rgb}{0,0,1}
\definecolor{green}{rgb}{0,1,0}
\definecolor{yellow}{rgb}{1,1,0}
\definecolor{orange}{rgb}{1,0.647,0}
\definecolor{gold}{rgb}{1,0.843,0}
\definecolor{purple}{rgb}{0.627,0.125,0.941}
\definecolor{gray}{rgb}{0.745,0.745,0.745}
\definecolor{brown}{rgb}{0.647,0.165,0.165}
\definecolor{navy}{rgb}{0,0,0.502}
\definecolor{pink}{rgb}{1,0.753,0.796}
\definecolor{seagreen}{rgb}{0.18,0.545,0.341}
\definecolor{turquoise}{rgb}{0.251,0.878,0.816}
\definecolor{violet}{rgb}{0.933,0.51,0.933}
\definecolor{darkblue}{rgb}{0,0,0.545}
\definecolor{darkcyan}{rgb}{0,0.545,0.545}
\definecolor{darkgray}{rgb}{0.663,0.663,0.663}
\definecolor{darkgreen}{rgb}{0,0.392,0}
\definecolor{darkmagenta}{rgb}{0.545,0,0.545}
\definecolor{darkorange}{rgb}{1,0.549,0}
\definecolor{darkred}{rgb}{0.545,0,0}
\definecolor{lightblue}{rgb}{0.678,0.847,0.902}
\definecolor{lightcyan}{rgb}{0.878,1,1}
\definecolor{lightgray}{rgb}{0.827,0.827,0.827}
\definecolor{lightgreen}{rgb}{0.565,0.933,0.565}
\definecolor{lightyellow}{rgb}{1,1,0.878}
\definecolor{black}{rgb}{0,0,0}
\definecolor{white}{rgb}{1,1,1}

\fi

\begin{document}
\allowdisplaybreaks

\twocolumn[{
\begin{frontmatter}
\title{The Maximum Linear Arrangement Problem for trees under projectivity and planarity}

\author[affiliation1]{Llu\'is Alemany-Puig\corref{cor1}}
\ead{lalemany@cs.upc.edu}
\cortext[cor1]{Corresponding author}

\author[affiliation1]{Juan Luis Esteban}
\ead{esteban@cs.upc.edu}

\author[affiliation1]{Ramon Ferrer-i-Cancho}
\ead{rferrer@cs.upc.edu}

\affiliation[affiliation1]{
	organization={Quantitative, Mathematical and Computational Linguistics Research Group, Computer Science Department, Universitat Polit\`ecnica de Catalunya},
	addressline={Jordi Girona 1-3},
	postcode={08034},
	city={Barcelona},
	country={Catalonia, Spain}
}

\begin{abstract}
A linear arrangement is a mapping $\arr$ from the $n$ vertices of a graph $G$ to $n$ distinct consecutive integers. Linear arrangements can be represented by drawing the vertices along a horizontal line and drawing the edges as semicircles above said line. In this setting, the length of an edge is defined as the absolute value of the difference between the positions of its two vertices in the arrangement, and the cost of an arrangement as the sum of all edge lengths. Here we study two variants of the Maximum Linear Arrangement problem ({\tt MaxLA}), which consists of finding an arrangement that maximizes the cost. In the {\tt planar} variant for free trees, vertices have to be arranged in such a way that there are no edge crossings. In the {\tt projective} variant for rooted trees, arrangements have to be planar and the root of the tree cannot be covered by any edge. In this paper we present algorithms that are linear in time and space to solve {\tt planar} and {\tt projective MaxLA} for trees. We also prove several properties of maximum projective and planar arrangements, and show that caterpillar trees maximize {\tt planar MaxLA} over all trees of a fixed size thereby generalizing a previous extremal result on trees.
\end{abstract}

\end{frontmatter}
}]

\section{Introduction}
\label{sec:introduction}

A linear arrangement $\arr$ of a graph $G=(V,E)$, $n=|V|$, is a mapping of the vertices in $V$ to distinct consecutive integers in $[1,n]$; it can be seen both as a linear ordering of the vertices or as a permutation of the vertices, where vertices lie on a horizontal line in integer positions. In such an arrangement, the distance between two vertices $u,v\in V$ can be defined as $\dl{u,v}=|\arr(u) - \arr(v)|$. For any edge $uv$, $\dl{u,v}$ represents the length of edge $uv$ in $\arr$. We define the cost of an arrangement $\arr$ as the sum of all edge lengths $\D{G}=\sum_{uv\in E} \dl{u,v}$ \cite{Goldberg1976a,Shiloach1979a,Chung1984a,Hochberg2003a,Gildea2007a}.

The Minimum Linear Arrangement problem ({\tt minLA}) consists of minimizing the cost over all $n!$ possible arrangements $\arr$ of the vertices of a graph $G$. We denote this minimum as $\Dmin[G]$. The problem is NP-Hard for arbitrary graphs \cite{Garey1976a}. Various polynomial-time algorithms are available to solve {\tt minLA} for trees \cite{Goldberg1976a,Shiloach1979a,Chung1984a}. The fastest, to the best of our knowledge, is due to Chung \cite{Chung1984a}. On the other hand, the Maximum Linear Arrangement problem ({\tt MaxLA}) consists of finding the maximum cost, denoted here as $\DMax[G]$ \cite{Even1975a,Hassin2001a,Nurse2018a}; it is also NP-Hard for arbitrary graphs \cite{Even1975a}. Applications of {\tt MaxLA} can be found in placement of obnoxious facilities \cite{Tamir1991a}, and also in statistical normalization of the sum of edge lengths in Quantitative Linguistics studies \cite{Tily2010a,Gulordava2016a,Ferrer2022a}.

There exist several variants of edge length optimization problems in linear arrangements \cite{Alemany2022a}; two of them are the {\em planar} and the {\em projective} variants. In the {\em planar} variant ({\tt minLA}/{\tt MaxLA} under the {\em planarity} constraint), the placement of the vertices of a free tree is constrained so that there are no edge crossings. Consider two undirected edges of a graph $st,uv\in E$. Assume, without loss of generality (w.l.o.g.), that $\arr(s)<\arr(t)$, $\arr(u)<\arr(v)$ and $\arr(s)<\arr(u)$; then $st$ and $uv$ cross if $\arr(s) < \arr(u) < \arr(t) < \arr(v)$ \cite{Bernhart1974a}. Arrangements without edge crossings are known as {\em planar} arrangements \cite{Kuhlmann2006a}, and also one-page book embeddings \cite{Bernhart1974a}. There are several $\bigO{n}$-time algorithms to solve {\tt planar minLA} \cite{Iordanskii1987a,Hochberg2003a,Alemany2022a}. In the {\em projective} variant ({\tt minLA}/{\tt MaxLA} under the projectivity constraint), a rooted tree is arranged so that the arrangement is planar and the root is not covered. Given an edge $uv$, assume w.l.o.g., $\arr(u)<\arr(v)$; then a vertex $w$ is covered by $uv$ if $\arr(u)<\arr(w)<\arr(v)$. Arrangements of a rooted tree without edge crossings where the root is not covered are known as {\em projective} \cite{Kuhlmann2006a,Melcuk1988a}. There are several $\bigO{n}$-time algorithms to solve {\tt projective minLA} \cite{Hochberg2003a,Gildea2007a,Alemany2022a}.

{\tt Projective minLA} can be solved in time $\bigO{n}$, by arranging the subtrees of the rooted tree on alternating sides of the root in the arrangement from smallest to largest in an inside-out fashion so that vertices of the same subtree are arranged contiguously \cite{Hochberg2003a,Gildea2007a,Alemany2022a} (Figure \ref{fig:introduction:shape_minimum_projective}). {\tt Planar minLA} can be solved in time $\bigO{n}$ \cite{Hochberg2003a,Alemany2022a} by first finding one of the two centroidal\footnote{See \cite[p. 35]{Harary1969a} for a definition of centroid of a tree.} vertices $c$ of the input free tree and then solving {\tt projective minLA} for the same tree but rooted at $c$ \cite{Hochberg2003a,Alemany2022a}.

\begin{figure}
	\centering
\scalebox{0.925}{
\begin{tikzpicture}[ipe stylesheet]
  \draw
    (64, 752) rectangle (128, 744);
  \node[ipe node]
     at (88, 732) {$\SubRtree{1}$};
  \draw
    (248, 752) rectangle (312, 744);
  \node[ipe node]
     at (272, 732) {$\SubRtree{2}$};
  \draw
    (160, 752) rectangle (176, 744);
  \draw
    (200, 752) rectangle (216, 744);
  \node[ipe node]
     at (156, 732) {$\SubRtree{k-1}$};
  \node[ipe node]
     at (204, 732) {$\SubRtree{k}$};
  \pic
     at (188, 748) {ipe disk};
  \pic[ipe mark small]
     at (136, 748) {ipe disk};
  \pic[ipe mark small]
     at (144, 748) {ipe disk};
  \pic[ipe mark small]
     at (152, 748) {ipe disk};
  \pic[ipe mark small]
     at (224, 748) {ipe disk};
  \pic[ipe mark small]
     at (232, 748) {ipe disk};
  \pic[ipe mark small]
     at (240, 748) {ipe disk};
  \node[ipe node]
     at (184, 736) {$\Root$};
  \draw[-{>[ipe arrow tiny]}]
    (187.9999, 747.9999)
     arc[start angle=30.4967, end angle=144.5243, radius=54.8917];
  \draw[-{>[ipe arrow tiny]}]
    (188.0001, 747.9999)
     arc[start angle=30.4967, end angle=144.5243, x radius=-54.8917, y radius=54.8917];
  \draw[-{>[ipe arrow tiny]}]
    (188, 748)
     arc[start angle=-148.5994, end angle=-54.0194, x radius=13.8786, y radius=-13.8786];
  \draw[-{>[ipe arrow tiny]}]
    (188, 748)
     arc[start angle=-148.5994, end angle=-54.0194, radius=-13.8786];
\end{tikzpicture}
}
	\caption{The basic layout of a minimum projective arrangement \cite{Hochberg2003a}. The subtrees are arranged on alternating sides of the root $\Root$, from smallest at the center of the arrangement to largest at the ends. Note that $|V(\SubRtree{1})|\geq\cdots\geq|V(\SubRtree{k})|$. In this figure, $k$ denotes the degree of $\Root$.}
	\label{fig:introduction:shape_minimum_projective}
\end{figure}
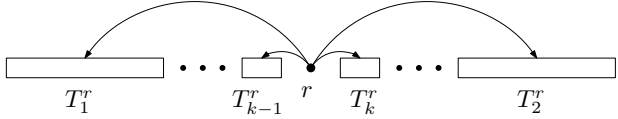

Much research has been carried out on the minimum variant of the problem. However, to the best of our knowledge, comparatively less research has been conducted on its maximum variant. It was probably first studied by Harper \cite{Harper1964a,Harper1966a}. On trees, some research has been carried out on {\tt unconstrained MaxLA}, denoted as $\DMax$, and {\tt planar MaxLA}, denoted as $\DMaxPlan$, where `$\mathrm{pl}$' stands for `planar', as well as on {\tt projective MaxLA}, denoted as $\DMaxProj$, where `$\mathrm{pr}$' stands for `projective' \cite{Ferrer2013a,Ferrer2015a,Ferrer2021a}. Concerning the unconstrained case, it has been shown that $\DMax$ is maximized by a balanced bistar (free) tree \cite{Ferrer2021a} (denoted as $\BalancedBistar$, illustrated in Figure \ref{fig:introduction:tree_types}(a)), i.e.,
\begin{align*}
\max_{\Ftree \in \FreeTrees} \{\DMax\}
	&= \DMax[ \BalancedBistar ] \\
	&= \frac{1}{4}\left[3(n - 1)^2 + 1 - n \bmod 2\right],
\end{align*}
where $\FreeTrees$ is the set of all $n$-vertex free trees. Exact values of $\DMax$ have been derived for the class of bistar trees (which comprises star and quasistar trees, shown in Figure \ref{fig:introduction:tree_types}(b,d)), and linear trees (Figure \ref{fig:introduction:tree_types}(c)) \cite{Ferrer2013a,Ferrer2021a}. In \cite{Nurse2018a,Nurse2019a}, Nurse {\em et. al.} proposed a $\bigO{n^{4d}}$-time algorithm to solve {\tt unconstrained MaxLA} on trees whose maximum degree is bounded by a constant $d$.

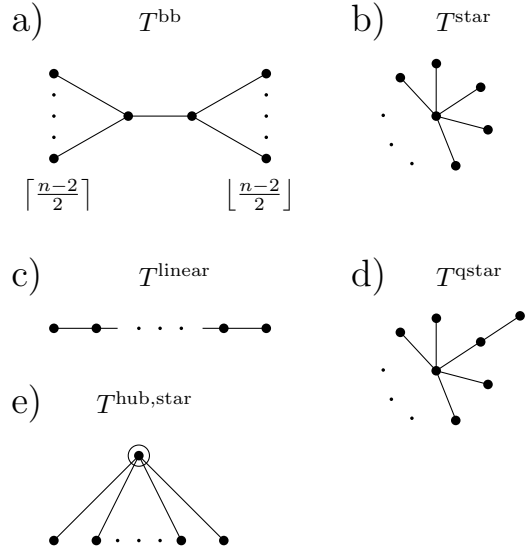
\begin{figure}
	\centering
\scalebox{1}{
\begin{tikzpicture}[ipe stylesheet]
  \pic
     at (96, 752) {ipe disk};
  \pic
     at (96, 720) {ipe disk};
  \pic
     at (124, 736) {ipe disk};
  \pic
     at (148, 736) {ipe disk};
  \pic
     at (176, 752) {ipe disk};
  \pic
     at (176, 720) {ipe disk};
  \draw
    (96, 752)
     -- (124, 736)
     -- (148, 736)
     -- (176, 752);
  \draw
    (148, 736)
     -- (176, 720);
  \draw
    (124, 736)
     -- (96, 720);
  \pic[ipe mark tiny]
     at (96, 744) {ipe disk};
  \pic[ipe mark tiny]
     at (96, 736) {ipe disk};
  \pic[ipe mark tiny]
     at (96, 728) {ipe disk};
  \pic[ipe mark tiny]
     at (176, 744) {ipe disk};
  \pic[ipe mark tiny]
     at (176, 736) {ipe disk};
  \pic[ipe mark tiny]
     at (176, 728) {ipe disk};
  \node[ipe node]
     at (84, 704) {$\left\lceil \frac{n-2}{2} \right\rceil$};
  \node[ipe node]
     at (160, 704) {$\left\lfloor \frac{n-2}{2} \right\rfloor$};
  \node[ipe node]
     at (128, 768) {$\BalancedBistar$};
  \pic
     at (240, 736) {ipe disk};
  \pic
     at (240.064, 755.791) {ipe disk};
  \pic
     at (256.7602, 746.8745) {ipe disk};
  \pic
     at (259.5113, 730.8426) {ipe disk};
  \pic
     at (247.3684, 717.2768) {ipe disk};
  \pic
     at (226.4988, 750.4787) {ipe disk};
  \draw
    (240, 736)
     -- (226.513, 750.547);
  \draw
    (240, 736)
     -- (240, 756);
  \draw
    (240, 736)
     -- (256.687, 746.963);
  \draw
    (240, 736)
     -- (259.369, 730.913);
  \draw
    (240, 736)
     -- (248, 716);
  \pic[ipe mark tiny]
     at (230.857, 718.105) {ipe disk};
  \pic[ipe mark tiny]
     at (223.1167, 725.2614) {ipe disk};
  \pic[ipe mark tiny]
     at (220.05, 736.2869) {ipe disk};
  \node[ipe node, font=\Large]
     at (80, 768) {a)};
  \node[ipe node, font=\Large]
     at (208, 768) {b)};
  \node[ipe node, font=\Large]
     at (208, 672) {d)};
  \pic
     at (96, 656) {ipe disk};
  \pic
     at (112, 656) {ipe disk};
  \pic
     at (160, 656) {ipe disk};
  \pic
     at (176, 656) {ipe disk};
  \pic[ipe mark tiny]
     at (128, 656) {ipe disk};
  \pic[ipe mark tiny]
     at (136, 656) {ipe disk};
  \pic[ipe mark tiny]
     at (144, 656) {ipe disk};
  \draw
    (96, 656)
     -- (120, 656);
  \draw
    (152, 656)
     -- (176, 656);
  \node[ipe node, font=\Large]
     at (80, 672) {c)};
  \pic
     at (240, 640) {ipe disk};
  \pic
     at (240.064, 659.791) {ipe disk};
  \pic
     at (256.7602, 650.874) {ipe disk};
  \pic
     at (259.5113, 634.843) {ipe disk};
  \pic
     at (247.3684, 621.277) {ipe disk};
  \pic
     at (226.4988, 654.479) {ipe disk};
  \draw
    (240, 640)
     -- (226.513, 654.547);
  \draw
    (240, 640)
     -- (240, 660);
  \draw
    (240, 640)
     -- (256.687, 650.963);
  \draw
    (240, 640)
     -- (259.369, 634.913);
  \draw
    (240, 640)
     -- (248, 620);
  \pic[ipe mark tiny]
     at (230.857, 622.105) {ipe disk};
  \pic[ipe mark tiny]
     at (223.1167, 629.2614) {ipe disk};
  \pic[ipe mark tiny]
     at (220.05, 640.2869) {ipe disk};
  \draw
    (271.689, 660.751)
     -- (256.691, 650.845);
  \pic
     at (271.6, 660.6993) {ipe disk};
  \node[ipe node]
     at (240, 768) {$\Star$};
  \node[ipe node]
     at (128, 672) {$\Linear$};
  \node[ipe node, font=\Large]
     at (80, 624) {e)};
  \node[ipe node]
     at (112, 624) {$\StarHub$};
  \pic
     at (128, 608) {ipe disk};
  \pic
     at (96, 576) {ipe disk};
  \pic
     at (112, 576) {ipe disk};
  \pic
     at (144, 576) {ipe disk};
  \pic
     at (160, 576) {ipe disk};
  \pic[ipe mark tiny]
     at (120, 576) {ipe disk};
  \pic[ipe mark tiny]
     at (128, 576) {ipe disk};
  \pic[ipe mark tiny]
     at (136, 576) {ipe disk};
  \draw
    (96, 576)
     -- (128, 608)
     -- (112, 576);
  \draw
    (144, 576)
     -- (128, 608)
     -- (160, 576);
  \draw
    (128, 608) circle[radius=4];
  \node[ipe node]
     at (240, 672) {$\QuasiStar$};
\end{tikzpicture}
}
	\caption{Illustration of several types of trees. a) Balanced bistar tree $\BalancedBistar$ of $n$ vertices. b) Star tree $\Star$. c) Linear tree $\Linear$. d) Quasi-star tree $\QuasiStar$ (a type of bistar tree). e) A star tree rooted at the hub $\StarHub$.}
	\label{fig:introduction:tree_types}
\end{figure}

Here we fill the gap in the literature on the calculation of {\tt MaxLA} for a given tree by devising algorithms that solve {\tt projective} and {\tt planar MaxLA} problems for trees in time and space $\bigO{n}$. Furthermore, we give the value of $\DMaxProj$ over all $\Rtree\in\RootedTrees$, where $\RootedTrees$ is the set of all $n$-vertex rooted trees, and generalize a previous extremal result in \cite{Ferrer2013a} by showing that the maximum value of $\DMaxPlan$ over all free trees $\Ftree\in\FreeTrees$ is achieved at least by caterpillar trees.

We begin by defining notation and concepts in Section \ref{sec:preliminaries}. In Section \ref{sec:outline}, we present an outline of an algorithm for trees unifying previous research on {\tt constrained minLA} and the new algorithms presented in this paper for solving {\tt constrained MaxLA}, thus highlighting common patterns. In Section \ref{sec:projective}, we show that {\tt projective MaxLA} can be solved in time and space $\bigO{n}$ using an strategy similar to that used to solve {\tt projective minLA}. We provide a $\bigO{n}$-time and $\bigO{n}$-space algorithm to solve {\tt planar MaxLA} in Section \ref{sec:planar}. Implementations of all algorithms presented here are available in the Linear Arrangement Library\footnote{Available publicly online at \url{https://github.com/LAL-project/linear-arrangement-library}.} \cite{Alemany2021a}.
\section{Preliminaries}
\label{sec:preliminaries}

We denote vertices as  $u,v,\Root,w,x,y,z$. We denote integers as $a,b,i,j,k,l$. We denote free trees as $\Ftree$ and a tree rooted at vertex $\Root$ as $\Rtree$; we consider the edges of a rooted tree to be oriented away from the root. We denote the degree of a vertex $u$ of a free tree $\Ftree$ as $\degree{u}$, equal to the number of neighbors, and the out-degree of a vertex of a rooted tree as $\outdegree{u}$, equal to the number of out-neighbors, or children, in a rooted tree $\Rtree$. Undirected edges are denoted as $uv$, and directed edges as $(u,v)$. We denote a subtree of $\Rtree$ rooted at $u\in V$ as $\SubRtree{u}$; subtree $\SubRtree{v}$ is an {\em immediate subtree} of $\SubRtree{u}$ if $v$ is a child of $u$. Let $\Nvert{v}{u}=|V(\SubRtree{u}[v])|$; it is easy to see that for any edge $uv$, $\Nvert{u}{v} + \Nvert{v}{u}=n$. We say that the {\em size of a child} $u$ of $\Root$ in $\Rtree$ is $\Nvert{\Root}{u}$. It will be useful to order the immediate subtrees of $\SubRtree{v}[u]$ according to size in a non-increasing manner. Let $\largestNvert{u}{v}{i}$ be the size of the $i$-th largest immediate subtree of $\SubRtree{v}[u]$. For the sake of brevity, let $\SubRtree{i}[u]$ be the $i$-th largest immediate subtree of $\Rtree[u]$.
\section{Outline of the algorithms}
\label{sec:outline}

The algorithms for all four problems {\tt projective}/{\tt planar} {\tt minLA}/{\tt MaxLA} on trees are all of a similar structure, which can be summarized with the following steps.
\begin{enumerate}
\item \label{step:outline:find_optimal_root} Find an optimal root $v$. In both projective variants, $v$ is the root of the tree. In {\tt planar minLA}, $v$ is a centroidal vertex \cite{Iordanskii1987a,Hochberg2003a,Alemany2022a}; there exists several algorithms to find it \cite{Goldman1971a,Hochberg2003a}\footnote{In Goldman's work \cite{Goldman1971a} there is an algorithm that can be adapted to compute the centroid of a tree and, at the same time, compute subtree sizes.}. In {\tt planar MaxLA}, a necessary condition for $v$ is to have a leaf attached (Lemma \ref{lem:planar:max_root_internal_with_leaves}); how to find $v$ is detailed in Algorithm \suppref{4.2} (Supplementary material).

\item \label{step:outline:subtree sizes} Calculate subtree sizes with respect to $v$.

\item \label{step:outline:top_down} In a top-down fashion, starting at $\Rtree[v]$,
	\begin{enumerate}
	\item Sort the immediate subtrees by their size. This step requires sorting in $\bigO{n}$; counting sort \cite{Cormen2001a} achieves this goal.
	
	\item Find an optimal position for the root of the current subtree. In {\tt planar minLA}, the immediate subtrees are arranged to both sides of the root in a balanced manner \cite{Hochberg2003a,Alemany2022a}. In {\tt planar MaxLA}, all the immediate subtrees are placed to the left (or to the right) of the root.
	
	\item Compute the interval $[a,b]$ in which each immediate subtree is to be arranged (Figure \ref{fig:introduction:shape_minimum_projective} depicts the arrangement in the minimization variants; Figure \ref{fig:projective:shape_maximum} depicts the arrangement in the maximization variants).
	
	\item Recursively apply step \ref{step:outline:top_down} to all immediate subtrees.
	\end{enumerate}
\end{enumerate}
\section{{\tt Projective MaxLA}}
\label{sec:projective}

It is easy to describe intuitively the shape of a maximum projective arrangement for $\Rtree$. Note that in order to avoid edge crossings the vertices of any subtree must be arranged in consecutive positions. The root of the subtree must be arranged on the leftmost (rightmost) position of the allowed interval to arrange the subtree. All the subtrees of the subtree being arranged must be arranged non-increasingly by size, the biggest one nearest the root and so on, maximizing the length of the edges starting at the root. For the initial tree we put the root, say, at the left side and the allowed interval is all the positions. This is a {\em right branching} arrangement; if the root is on the rightmost position, then it is a {\em left branching} arrangement. A glance at Figure \ref{fig:projective:shape_maximum} will dispel all remaining doubts. In Theorem \ref{thm:projective:shape_maximum}, we prove that the arrangement described is indeed  a maximum projective arrangement of a rooted tree.

\begin{theorem}
\label{thm:projective:shape_maximum}
A maximum projective arrangement of a rooted tree $\Rtree$ is such that
\begin{enumerate}[(i)]

\item \label{thm:projective:shape_maximum:3} the immediate subtrees of $\SubRtree{u}$ are arranged non-increasingly by size, with the smallest one farthest from $u$, all at the same side of $u$ in the arrangement, and

\item \label{thm:projective:shape_maximum:1} for every directed edge $(u,v)$, if the arrangement of $\SubRtree{u}$ is left (resp. right) branching then the arrangement of $\SubRtree{v}$ is right (resp. left)  branching.

\end{enumerate}
\end{theorem}
\begin{proof}
First, we prove (\ref{thm:projective:shape_maximum:3}) with a strategy similar to that by Hochberg and Stallmann for the minimization problem \cite[Lemma 6]{Hochberg2003a}. Consider a rooted tree $\Rtree$ and a projective arrangement of $\Rtree$ constructed from a permutation of its immediate subtrees and the root $\Root$. Note that projective (and planar) arrangements have to be constructed in this way, otherwise there would be crossings; thus the arrangements of all immediate subtrees are constructed likewise, down to the leaves. W.l.o.g., assume that in the permutation there are more (or equal number of) vertices to the right of $\Root$ than to the left of $\Root$.

We define two steps which, when applied to any such arrangement, the cost never decreases. {\em Step A} consists of swapping two subtrees in the permutation. For two trees to the right (or to the left) of the root, if the smaller one is nearer the root, we swap their positions in the permutation. This always increases the total cost of the arrangement. {\em Step B} consists of moving the leftmost subtree in the arrangement to the right side of the arrangement, as far as possible from the root. This increases the cost of the arrangement since we assumed, w.l.o.g., that there are more (or equal number of) vertices to the right of the tree's root than to its left.

The procedure to construct a maximum projective arrangement is as follows. {\em Step A} is applied to both sides of the arrangement of $\Rtree$, as often as needed until all subtrees are eventually arranged non-increasingly by size with the largest subtree closest to the root, and the smallest subtree farthest from it. Then, {\em step B} is applied until all subtrees to the left of the root have been placed to its right. Lastly, apply {\em step A} again as many times as needed to obtain a non-increasing order of all the subtrees with the largest subtree next to the root. These steps construct a right branching arrangement, but they are easily modifiable to construct a left branching arrangement by mirroring. Note that the order of two equally large subtrees in the permutation of intervals does not matter, since swapping them does not change the local cost of the subtrees, and the length of the edges from their roots to their root's parent does not change.

Second, we prove (\ref{thm:projective:shape_maximum:1}). We apply recursively the procedure above to the subtrees, alternating between right branching arrangement and left branching in order to maximize the cost. If a subtree has a left branching arrangement, each of its immediate subtrees $\SubRtree{u}$ must have a right branching arrangement and vice versa so as to maximize the length of the edge $(\Root, u)$.
\end{proof}

Now we list several properties that characterize maximum projective arrangements, all of which follow immediately from Theorem \ref{thm:projective:shape_maximum}.

\begin{corollary}
\label{cor:projective:characterizing_maximum_arrangements}
In a maximum projective arrangement $\arr$ of $\Rtree$,
\begin{enumerate}[(i)]
\item \label{cor:projective:characterizing_maximum_arrangements:i} The root $\Root$ is placed at one end of the linear arrangement, 
\item \label{cor:projective:characterizing_maximum_arrangements:ii} All leaves of $r$, if any, are arranged consecutively at the other end of the linear arrangement,
\item \label{cor:projective:characterizing_maximum_arrangements:iii} The first and last vertices of the arrangement are adjacent in the tree,
\item \label{cor:projective:characterizing_maximum_arrangements:iv} The arrangements of all immediate subtrees of any vertex branch towards the same direction.
\end{enumerate}
\end{corollary}

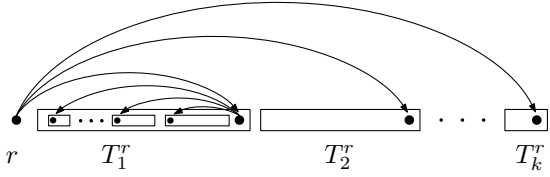
\begin{figure}
	\centering
\scalebox{1}{
\begin{tikzpicture}[ipe stylesheet]
  \pic
     at (69.5363, 700.4051) {ipe disk};
  \draw
    (77.5363, 704.4051) rectangle (157.5363, 696.4051);
  \draw
    (161.5363, 704.4051) rectangle (221.5363, 696.4051);
  \pic[ipe mark tiny]
     at (229.5363, 700.4051) {ipe disk};
  \pic[ipe mark tiny]
     at (237.5363, 700.4051) {ipe disk};
  \pic[ipe mark tiny]
     at (245.5363, 700.4051) {ipe disk};
  \draw
    (253.5363, 704.4051) rectangle (269.5363, 696.4051);
  \pic
     at (153.5363, 700.4051) {ipe disk};
  \pic
     at (217.5363, 700.4051) {ipe disk};
  \pic
     at (265.5363, 700.4051) {ipe disk};
  \draw[-{>[ipe arrow tiny]}]
    (69.5363, 702.1292)
     arc[start angle=-159.146, end angle=-24.1138, x radius=45.2236, y radius=-25.0485];
  \draw[-{>[ipe arrow tiny]}]
    (69.5365, 702.1292)
     arc[start angle=-162.0307, end angle=-19.7214, x radius=77.8228, y radius=-43.1046];
  \draw[-{>[ipe arrow tiny]}]
    (69.5364, 702.1292)
     arc[start angle=-166.239, end angle=-15.1744, x radius=100.966, y radius=-55.9233];
  \draw
    (125.6642, 702.2202) rectangle (149.6642, 698.2202);
  \draw
    (81.6642, 702.2202) rectangle (89.6642, 698.2202);
  \draw
    (105.6642, 702.2202) rectangle (121.6642, 698.2202);
  \pic[ipe mark tiny]
     at (93.7091, 700.1511) {ipe disk};
  \pic[ipe mark tiny]
     at (97.7091, 700.1511) {ipe disk};
  \pic[ipe mark tiny]
     at (101.7091, 700.1511) {ipe disk};
  \pic[ipe mark small]
     at (83.3593, 700.3178) {ipe disk};
  \pic[ipe mark small]
     at (107.5353, 700.2481) {ipe disk};
  \pic[ipe mark small]
     at (127.5373, 700.2351) {ipe disk};
  \draw[-{>[ipe arrow tiny]}]
    (153.5363, 702.1292)
     arc[start angle=43.3055, end angle=131.6209, x radius=17.6472, y radius=9.7745];
  \draw[-{>[ipe arrow tiny]}]
    (153.5367, 702.1294)
     arc[start angle=34.9791, end angle=141.6537, x radius=27.8893, y radius=15.4474];
  \draw
    (153.5363, 702.1292)
     arc[start angle=0, end angle=0, x radius=32, y radius=17.7242];
  \draw[-{>[ipe arrow tiny]}]
    (153.5362, 702.1292)
     arc[start angle=29.7693, end angle=147.9949, x radius=40.2288, y radius=22.282];
  \node[ipe node]
     at (65.536, 684.405) {$\Root$};
  \node[ipe node]
     at (101.536, 684.405) {$\SubRtree{1}$};
  \node[ipe node]
     at (185.536, 684.405) {$\SubRtree{2}$};
  \node[ipe node]
     at (257.536, 684.405) {$\SubRtree{k}$};
\end{tikzpicture}
}
	\caption{A maximum projective arrangement of $\Rtree$. The arrangement of $\Rtree$ is right branching, while the arrangements of $\SubRtree{1}, \dots, \SubRtree{k}$ are left branching.}
	\label{fig:projective:shape_maximum}
\end{figure}

The cost of the maximum projective arrangement of $\Rtree$ is, looking at Figure \ref{fig:projective:shape_maximum}, easy to derive. It is the sum of the maximum projective arrangements of the subtrees plus the sum of the lengths of the $k$ edges incident to the root to its children: the length of the first edge is the size of $\SubRtree{1}$, the length of the second edge is the sum of the sizes of $\SubRtree{1}$ and $\SubRtree{2}$, and so on. Corollary \ref{cor:projective:cost_max_pr_arr} formalizes this cost.

\begin{corollary}
\label{cor:projective:cost_max_pr_arr}
For any rooted tree $\Rtree$,
\begin{equation*}
\DMaxProj[\Rtree]
	=
	\sum_{i=1}^{\degree{\Root}} \DMaxProj[\SubRtree{i}]
		 + \sum_{i=1}^{\degree{\Root}}\sum_{j=1}^i \largestNvert{\Root}{\Root}{j}.
\end{equation*}
\end{corollary}

It is easy to devise an algorithm to solve {\tt projective MaxLA} from the description of the maximum arrangement and the formal proof in Theorem \ref{thm:projective:shape_maximum}. The algorithm can be viewed as a modification of an algorithm previously used to solve {\tt projective minLA} \cite[Algorithm 4.1]{Alemany2022a} where the intervals assigned to the subtrees are calculated differently.

\begin{corollary}
\label{cor:projective:algorithm}
There is an algorithm that, for any tree $\Rtree$, solves {\tt projective MaxLA} in time and space $\bigO{n}$, giving one of the possibly many maximum projective arrangements.
\end{corollary}
\begin{proof}
We detail such algorithm and give its proof of correctness and cost in the Supplementary material (Theorem \suppref{3.1}, Algorithm \suppref{3.1}).
\end{proof}

We finish this section with the maximum value of $\DMaxProj$ over all rooted trees $\Rtree\in\RootedTrees$. We show that it can be achieved by $\StarHub$ (Figure \ref{fig:introduction:tree_types}(e)), a star tree rooted at the vertex of maximum degree (the so-called {\em hub}).
\begin{property}
\label{proper:introduction:max_rooted_trees}
\begin{equation*}
\max_{\Rtree \in \RootedTrees} \{\DMaxProj\} = \DMaxProj[\StarHub] = {n \choose 2}.
\end{equation*}
\end{property}
\begin{proof}
The statement follows from combining two facts. First, it was shown that \cite{Ferrer2013a},
\begin{equation*}
\max_{\Ftree \in \FreeTrees} \{\DMaxPlan\} = {n \choose 2}.
\end{equation*}
And, since planarity is a generalization of projectivity, we have that
\begin{equation*}
\max_{\Rtree \in \RootedTrees} \{\DMaxProj\}
	\leq \max_{\Ftree \in \FreeTrees} \{\DMaxPlan\}
	 = {n \choose 2}.
\end{equation*}
Second, trivially \cite{Ferrer2015a}
\begin{equation*}
\DMaxProj[\StarHub] = \DMaxPlan[\Star] = {n \choose 2}.
\end{equation*}
Bear in mind that all linear arrangements of a star tree are planar because all pairs of edges share a vertex and thus they cannot cross, and besides, in a star tree rooted at the hub vertex, all linear arrangements are projective because the root cannot be covered.
\end{proof}
\section{{\tt Planar MaxLA}}
\label{sec:planar}

We first show that a maximum planar arrangement of a free tree is in fact a maximum projective arrangement for the same tree rooted at a certain vertex.
\begin{theorem}
\label{thm:planar:quadratic}
Given a free tree $\Ftree$, a maximum planar arrangement for $\Ftree$ is also a maximum projective arrangement for $\Rtree[u]$ for some $u\in V$. Formally,
\begin{equation*}
\DMaxPlan = \max_{u\in V} \{ \DMaxProj[\Rtree[u]] \}.
\end{equation*}
\end{theorem}
\begin{proof}
Let $\Pi_{\planar}(\Ftree,u)$ be the set of planar arrangements $\arr$ of $\Ftree$ such that $\arr(u)=1$, and let $\Pi_{\projective}(\Rtree[v],v)$ be the set of all projective arrangements $\arr'$ of $\Rtree[v]$ such that $\arr'(v)=1$. Notice that $\Pi_{\planar}(\Ftree,u)=\Pi_{\projective}(\Rtree[u],u)$ (see \cite{Alemany2022b} for further details on the relationship between planar arrangements and projective arrangements). Then it is easy to see that
\begin{align*}
\DMaxPlan
	&= \max_{ u\in V } \max_{ \arr\in\Pi_{\planar}(\Ftree,u) } \{ \D{\Ftree} \} \\
	&= \max_{ u\in V } \max_{ \arr\in\Pi_{\projective}(\Rtree[u],u) } \{ \D{\Rtree[u]} \} \\
	&= \max_{ u\in V } \{ \DMaxProj[\Rtree[u]] \}.
\end{align*}
\end{proof}

From Theorem \ref{thm:planar:quadratic}, we can calculate the maximum planar arrangement by rooting the tree in every vertex and keeping the arrangement that yields the maximum cost. It is easy to see that this strategy has time complexity $\bigO{n^2}$ for any tree. Later in this section we devise a $\bigO{n}$-time algorithm that follows the outline in Section \ref{sec:outline}. In order to achieve this, we first prove that calculating the cost of a maximum projective arrangement of a tree rooted at $v$ is relatively easy if we know the cost of the maximum projective arrangement for the same tree rooted at a neighbor of $v$.

\begin{lemma}
\label{lem:planar:relation_neighboring_vertices}
Let $\Ftree$ be a free tree. For any edge $uv\in E(\Ftree)$, it holds that
\begin{equation}
\label{eq:planar:relation_neighboring_vertices}
\DMaxProj[\Rtree[v]] - \DMaxProj[\Rtree[u]] = f(v,u) - f(u,v),
\end{equation}
where
\begin{equation*}
f(u,v) = [\degree{u} - j]\Nvert{u}{v} + \sum_{i=1}^{j} \singlelargestNvert{u}{i}
\end{equation*}
and $j$, $1\le j\le\degree{u}$, is the position of $v$ in the list of immediate subtrees of $\Rtree[u]$ sorted non-increasingly by size.
\end{lemma}
\begin{proof}
Let $k$ (resp. $l$) be the index of $v$ (resp. $u$) in the sorted list of neighbors of $u$ (resp. $v$), as depicted in Figures \ref{fig:planar:relation_neighboring_vertices}(a,b). Now, it is easy to see from Figure \ref{fig:planar:relation_neighboring_vertices}(c) that, for any edge $uv$, we can decompose $\DMaxProj[\Rtree[u]]$ into four parts: (blue) $\DMaxProj[\Rtree[u]\setminus\SubRtree{v}[u]]$, where $\Rtree[u]\setminus \SubRtree{v}[u]$ denotes the tree $\Rtree[u]$ without its immediate subtree $\SubRtree{v}[u]$, (red) $\DMaxProj[\SubRtree{v}[u]]$, (orange) the length of the edge from $u$ to $v$, and (green) the length of the segment of the edges from $u$ to $u_{k+1}$, \dots, $u_{\degree{u}}$ that cover the subtree $\SubRtree{v}[u]$ (which is not counted in $\DMaxProj[\Rtree[u]\setminus\SubRtree{v}[u]]$) and which is needed in order to insert $\SubRtree{v}[u]$ in the maximum projective arrangement of $\Rtree[u]\setminus\SubRtree{v}[u]$. More formally and in the same order,
\begin{align*}
\DMaxProj[\Rtree[u]]
	&= \DMaxProj[\Rtree[u]\setminus\SubRtree{v}[u]] + \DMaxProj[\SubRtree{v}[u]] \\
	&+ \sum_{i=1}^{k} \singlelargestNvert{u}{i} + (\degree{u} - k)\Nvert{u}{v}.
\end{align*}
Likewise for $\DMaxProj[\Rtree[v]]$. Now, since $\Rtree[u]\setminus\SubRtree{v}[u] = \SubRtree{u}[v]$ and $\Rtree[v]\setminus\SubRtree{u}[v] = \SubRtree{v}[u]$ (Figure \ref{fig:planar:relation_neighboring_vertices}(a,b)), it is easy to see that
\begin{equation*}
\DMaxProj[\Rtree[v]\setminus\SubRtree{u}[v]] + \DMaxProj[\SubRtree{u}[v]] = 
\DMaxProj[\Rtree[u]\setminus\SubRtree{v}[u]] + \DMaxProj[\SubRtree{v}[u]].
\end{equation*}
Hence Equation \ref{eq:planar:relation_neighboring_vertices}.
\end{proof}

\begin{figure*}
	\centering
\scalebox{1}{
\begin{tikzpicture}[ipe stylesheet]
  \pic
     at (144, 604) {ipe disk};
  \node[ipe node]
     at (132, 612) {$u=v_l$};
  \node[ipe node]
     at (44, 572) {$u_1$};
  \node[ipe node]
     at (96, 568) {$u_{k-1}$};
  \node[ipe node]
     at (176, 568) {$u_{k+1}$};
  \node[ipe node]
     at (232, 576) {$u_{\degree{u}}$};
  \pic
     at (60, 572) {ipe disk};
  \pic
     at (92, 572) {ipe disk};
  \pic
     at (204, 572) {ipe disk};
  \pic
     at (236, 572) {ipe disk};
  \filldraw[pattern=north east lines, pattern color=blue]
    (60, 572)
     -- (48, 524)
     -- (72, 524)
     -- cycle;
  \filldraw[pattern=north east lines, pattern color=blue]
    (92, 572)
     -- (80, 524)
     -- (104, 524)
     -- cycle;
  \filldraw[pattern=north east lines, pattern color=blue]
    (204, 572)
     -- (192, 524)
     -- (216, 524)
     -- cycle;
  \filldraw[pattern=north east lines, pattern color=blue]
    (236, 572)
     -- (224, 524)
     -- (248, 524)
     -- cycle;
  \node[ipe node]
     at (145.598, 578.719) {$u_k=v$};
  \draw
    (144, 604)
     -- (60, 572);
  \draw
    (144, 604)
     -- (92, 572);
  \draw
    (144, 604)
     -- (204, 572);
  \draw
    (144, 604)
     -- (236, 572);
  \pic
     at (144, 572) {ipe disk};
  \node[ipe node]
     at (84, 500) {$v_1$};
  \node[ipe node]
     at (108, 500) {$v_{l-1}$};
  \pic
     at (96, 496) {ipe disk};
  \pic
     at (128, 496) {ipe disk};
  \filldraw[pattern=north west lines, pattern color=red]
    (96, 496)
     -- (84, 448)
     -- (108, 448)
     -- cycle;
  \filldraw[pattern=north west lines, pattern color=red]
    (128, 496)
     -- (116, 448)
     -- (140, 448)
     -- cycle;
  \pic[ipe mark tiny]
     at (72, 572) {ipe disk};
  \pic[ipe mark tiny]
     at (76, 572) {ipe disk};
  \pic[ipe mark tiny]
     at (80, 572) {ipe disk};
  \pic[ipe mark tiny]
     at (216, 572) {ipe disk};
  \pic[ipe mark tiny]
     at (220, 572) {ipe disk};
  \pic[ipe mark tiny]
     at (224, 572) {ipe disk};
  \pic[ipe mark tiny]
     at (108, 496) {ipe disk};
  \pic[ipe mark tiny]
     at (112, 496) {ipe disk};
  \pic[ipe mark tiny]
     at (116, 496) {ipe disk};
  \node[ipe node]
     at (140, 500) {$v_{l+1}$};
  \node[ipe node]
     at (192, 500) {$v_{\degree{v}}$};
  \pic
     at (160, 496) {ipe disk};
  \pic
     at (192, 496) {ipe disk};
  \filldraw[pattern=north west lines, pattern color=red]
    (160, 496)
     -- (148, 448)
     -- (172, 448)
     -- cycle;
  \filldraw[pattern=north west lines, pattern color=red]
    (192, 496)
     -- (180, 448)
     -- (204, 448)
     -- cycle;
  \pic[ipe mark tiny]
     at (172, 496) {ipe disk};
  \pic[ipe mark tiny]
     at (176, 496) {ipe disk};
  \pic[ipe mark tiny]
     at (180, 496) {ipe disk};
  \draw
    (144, 604)
     -- (144, 572);
  \draw
    (144, 572)
     -- (96, 496);
  \draw
    (144, 572)
     -- (128, 496);
  \draw
    (144, 572)
     -- (160, 496);
  \draw
    (144, 572)
     -- (192, 496);
  \pic
     at (392, 604) {ipe disk};
  \node[ipe node]
     at (393.278, 578.241) {$v_l=u$};
  \node[ipe node]
     at (332, 500) {$u_1$};
  \node[ipe node]
     at (352, 500) {$u_{k-1}$};
  \node[ipe node]
     at (384, 500) {$u_{k+1}$};
  \node[ipe node]
     at (440, 500) {$u_{\degree{u}}$};
  \pic
     at (308, 572) {ipe disk};
  \pic
     at (340, 572) {ipe disk};
  \pic
     at (452, 572) {ipe disk};
  \pic
     at (484, 572) {ipe disk};
  \filldraw[pattern=north west lines, pattern color=red]
    (308, 572)
     -- (296, 524)
     -- (320, 524)
     -- cycle;
  \filldraw[pattern=north west lines, pattern color=red]
    (340, 572)
     -- (328, 524)
     -- (352, 524)
     -- cycle;
  \filldraw[pattern=north west lines, pattern color=red]
    (452, 572)
     -- (440, 524)
     -- (464, 524)
     -- cycle;
  \filldraw[pattern=north west lines, pattern color=red]
    (484, 572)
     -- (472, 524)
     -- (496, 524)
     -- cycle;
  \node[ipe node]
     at (380, 612) {$v=u_k$};
  \draw
    (392, 604)
     -- (308, 572);
  \draw
    (392, 604)
     -- (340, 572);
  \draw
    (392, 604)
     -- (452, 572);
  \draw
    (392, 604)
     -- (484, 572);
  \pic
     at (392, 572) {ipe disk};
  \node[ipe node]
     at (292, 572) {$v_1$};
  \node[ipe node]
     at (348, 568) {$v_{l-1}$};
  \pic
     at (344, 496) {ipe disk};
  \pic
     at (376, 496) {ipe disk};
  \filldraw[pattern=north east lines, pattern color=blue]
    (344, 496)
     -- (332, 448)
     -- (356, 448)
     -- cycle;
  \filldraw[pattern=north east lines, pattern color=blue]
    (376, 496)
     -- (364, 448)
     -- (388, 448)
     -- cycle;
  \pic[ipe mark tiny]
     at (320, 572) {ipe disk};
  \pic[ipe mark tiny]
     at (324, 572) {ipe disk};
  \pic[ipe mark tiny]
     at (328, 572) {ipe disk};
  \pic[ipe mark tiny]
     at (464, 572) {ipe disk};
  \pic[ipe mark tiny]
     at (468, 572) {ipe disk};
  \pic[ipe mark tiny]
     at (472, 572) {ipe disk};
  \pic[ipe mark tiny]
     at (356, 496) {ipe disk};
  \pic[ipe mark tiny]
     at (360, 496) {ipe disk};
  \pic[ipe mark tiny]
     at (364, 496) {ipe disk};
  \node[ipe node]
     at (428, 568) {$v_{l+1}$};
  \node[ipe node]
     at (480, 576) {$v_{\degree{v}}$};
  \pic
     at (408, 496) {ipe disk};
  \pic
     at (440, 496) {ipe disk};
  \filldraw[pattern=north east lines, pattern color=blue]
    (408, 496)
     -- (396, 448)
     -- (420, 448)
     -- cycle;
  \filldraw[pattern=north east lines, pattern color=blue]
    (440, 496)
     -- (428, 448)
     -- (452, 448)
     -- cycle;
  \pic[ipe mark tiny]
     at (420, 496) {ipe disk};
  \pic[ipe mark tiny]
     at (424, 496) {ipe disk};
  \pic[ipe mark tiny]
     at (428, 496) {ipe disk};
  \draw
    (392, 604)
     -- (392, 572);
  \draw
    (392, 572)
     -- (344, 496);
  \draw
    (392, 572)
     -- (376, 496);
  \draw
    (392, 572)
     -- (408, 496);
  \draw
    (392, 572)
     -- (440, 496);
  \node[ipe node, font=\Large]
     at (44, 608) {a)};
  \node[ipe node, font=\Large]
     at (292, 608) {b)};
  \pic[ipe mark tiny]
     at (280, 376) {ipe disk};
  \pic[ipe mark tiny]
     at (284, 376) {ipe disk};
  \pic[ipe mark tiny]
     at (288, 376) {ipe disk};
  \filldraw[pattern=north west lines, pattern color=red]
    (252, 378) rectangle (276, 374);
  \pic
     at (48, 376) {ipe disk};
  \pic
     at (392, 376) {ipe disk};
  \draw
    (250, 380) rectangle (398, 372);
  \pic[ipe mark tiny]
     at (460, 376) {ipe disk};
  \pic[ipe mark tiny]
     at (464, 376) {ipe disk};
  \pic[ipe mark tiny]
     at (468, 376) {ipe disk};
  \pic[ipe mark tiny]
     at (148, 376) {ipe disk};
  \pic[ipe mark tiny]
     at (152, 376) {ipe disk};
  \pic[ipe mark tiny]
     at (156, 376) {ipe disk};
  \filldraw[pattern=north east lines, pattern color=blue]
    (160, 380) rectangle (248, 372);
  \filldraw[pattern=north east lines, pattern color=blue]
    (400, 380) rectangle (456, 372);
  \filldraw[pattern=north east lines, pattern color=blue]
    (472, 380) rectangle (516, 372);
  \node[ipe node]
     at (132, 356) {$u_1$};
  \node[ipe node]
     at (228, 356) {$u_{k-1}$};
  \node[ipe node]
     at (436, 356) {$u_{k+1}$};
  \node[ipe node]
     at (494, 356) {$u_{\degree{u}}$};
  \node[ipe node]
     at (252, 366) {$v_1$};
  \node[ipe node]
     at (292, 366) {$v_{l-1}$};
  \node[ipe node]
     at (320, 366) {$v_{l+1}$};
  \node[ipe node]
     at (360, 366) {$v_{\degree{v}}$};
  \filldraw[pattern=north west lines, pattern color=red]
    (292, 378) rectangle (316, 374);
  \pic[ipe mark tiny]
     at (348, 376) {ipe disk};
  \pic[ipe mark tiny]
     at (352, 376) {ipe disk};
  \pic[ipe mark tiny]
     at (356, 376) {ipe disk};
  \pic[ipe mark small]
     at (256, 376) {ipe disk};
  \pic[ipe mark small]
     at (296, 376) {ipe disk};
  \filldraw[pattern=north west lines, pattern color=red]
    (320, 378) rectangle (344, 374);
  \pic[ipe mark small]
     at (324, 376) {ipe disk};
  \filldraw[pattern=north west lines, pattern color=red]
    (360, 378) rectangle (384, 374);
  \pic[ipe mark small]
     at (364, 376) {ipe disk};
  \filldraw[pattern=north east lines, pattern color=blue]
    (56, 380) rectangle (144, 372);
  \draw[blue, -{ipe pointed[ipe arrow small]}]
    (48, 376)
     arc[start angle=180, end angle=347.0272, x radius=43.5558, y radius=-17.8183];
  \draw[blue, -{ipe pointed[ipe arrow small]}]
    (47.9995, 376)
     arc[start angle=-180, end angle=-8.4093, x radius=95.5135, y radius=-27.3516];
  \draw[orange, -{ipe pointed[ipe arrow small]}]
    (48, 376)
     arc[start angle=180, end angle=351.2804, x radius=171.994, y radius=-39.5782];
  \draw[blue, -{ipe pointed[ipe arrow small]}]
    (48, 376)
     arc[start angle=-180, end angle=-4.8866, x radius=199.362, y radius=-46.9571];
  \draw[blue, -{ipe pointed[ipe arrow small]}]
    (48, 376)
     arc[start angle=180, end angle=355.808, x radius=229.307, y radius=-54.7211];
  \draw[red, -{ipe pointed[ipe arrow tiny]}]
    (391.9997, 376)
     arc[start angle=0, end angle=160.6188, x radius=13.7887, y radius=5.9094];
  \draw[red, -{ipe pointed[ipe arrow tiny]}]
    (392, 376)
     arc[start angle=0, end angle=167.0872, x radius=33.726, y radius=8.6724];
  \draw[red, -{ipe pointed[ipe arrow tiny]}]
    (392.0001, 376)
     arc[start angle=0, end angle=170.5018, x radius=47.7931, y radius=12.2897];
  \draw[red, -{ipe pointed[ipe arrow tiny]}]
    (391.9996, 376)
     arc[start angle=-0.2009, end angle=173.4001, x radius=67.671, y radius=17.4011];
  \node[ipe node]
     at (390, 366) {$v$};
  \node[ipe node]
     at (44, 356) {$u$};
  \node[ipe node]
     at (388, 356) {$u_k$};
  \node[ipe node, font=\Large]
     at (44, 432) {c)};
  \pic
     at (136, 376) {ipe disk};
  \pic
     at (240, 376) {ipe disk};
  \pic
     at (448, 376) {ipe disk};
  \pic
     at (508, 376) {ipe disk};
  \draw[ipe dash dashed]
    (250, 368)
     -- (250, 438);
  \draw[ipe dash dashed]
    (398, 368)
     -- (398, 438);
  \draw[green, ipe pen fat]
    (249.992, 422.928)
     .. controls (257.106, 422.739) and (260.5215, 422.6915) .. (263.7154, 422.6361)
     .. controls (266.9093, 422.5807) and (269.8817, 422.5173) .. (272.759, 422.4383)
     .. controls (275.6363, 422.3593) and (278.4187, 422.2647) .. (281.1222, 422.154)
     .. controls (283.8257, 422.0433) and (286.4503, 421.9167) .. (288.9168, 421.7902)
     .. controls (291.3833, 421.6637) and (293.6917, 421.5373) .. (296.4067, 421.3908)
     .. controls (299.1217, 421.2443) and (302.2433, 421.0777) .. (305.2093, 420.9007)
     .. controls (308.1753, 420.7237) and (310.9857, 420.5363) .. (314.326, 420.2402)
     .. controls (317.6663, 419.944) and (321.5367, 419.539) .. (324.9563, 419.2013)
     .. controls (328.376, 418.8637) and (331.345, 418.5933) .. (333.873, 418.3083)
     .. controls (336.401, 418.0233) and (338.488, 417.7237) .. (340.8138, 417.4162)
     .. controls (343.1397, 417.1087) and (345.7043, 416.7933) .. (348.1723, 416.4617)
     .. controls (350.6403, 416.13) and (353.0117, 415.782) .. (355.4762, 415.4495)
     .. controls (357.9407, 415.117) and (360.4983, 414.8) .. (362.9453, 414.4252)
     .. controls (365.3923, 414.0503) and (367.7287, 413.6177) .. (369.762, 413.2788)
     .. controls (371.7953, 412.94) and (373.5257, 412.695) .. (375.22, 412.385)
     .. controls (376.9143, 412.075) and (378.5727, 411.7) .. (380.3753, 411.3322)
     .. controls (382.178, 410.9643) and (384.125, 410.6037) .. (386.6993, 410.0341)
     .. controls (389.2735, 409.4645) and (392.475, 408.686) .. (397.623, 407.518);
  \draw[green, ipe pen fat]
    (249.988, 430.389)
     .. controls (258.146, 430.673) and (262.889, 430.7205) .. (266.9207, 430.7443)
     .. controls (270.9523, 430.768) and (274.2727, 430.768) .. (277.2608, 430.7522)
     .. controls (280.249, 430.7363) and (282.905, 430.7047) .. (285.7825, 430.6732)
     .. controls (288.66, 430.6417) and (291.759, 430.6103) .. (294.9368, 430.5472)
     .. controls (298.1147, 430.484) and (301.3713, 430.389) .. (303.7428, 430.3098)
     .. controls (306.1143, 430.2307) and (307.6007, 430.1673) .. (309.5928, 430.1198)
     .. controls (311.585, 430.0723) and (314.083, 430.0407) .. (316.5178, 429.9458)
     .. controls (318.9527, 429.851) and (321.3243, 429.693) .. (324.107, 429.5192)
     .. controls (326.8897, 429.3453) and (330.0833, 429.1557) .. (333.0398, 428.9975)
     .. controls (335.9963, 428.8393) and (338.7157, 428.7127) .. (341.2453, 428.5387)
     .. controls (343.775, 428.3647) and (346.115, 428.1433) .. (348.4548, 427.9378)
     .. controls (350.7947, 427.7323) and (353.1343, 427.5427) .. (355.7113, 427.3057)
     .. controls (358.2883, 427.0687) and (361.1027, 426.7843) .. (363.766, 426.5828)
     .. controls (366.4293, 426.3813) and (368.9417, 426.2627) .. (371.6673, 426.041)
     .. controls (374.393, 425.8193) and (377.332, 425.4947) .. (379.9423, 425.1443)
     .. controls (382.5527, 424.794) and (384.8343, 424.418) .. (387.4362, 424.0763)
     .. controls (390.038, 423.7345) and (392.96, 423.427) .. (397.677, 422.735);
\end{tikzpicture}
}
	\caption{Proof of Lemma \ref{lem:planar:relation_neighboring_vertices}. a) A free tree $\Ftree$ rooted at $u$. b) A free tree $\Ftree$ rooted at $v$. c) A maximum projective arrangement of $\Rtree[u]$. The vertices $u$ and $v$ are adjacent in the tree, that is $uv\in E$. The children of both $u$ and $v$ are ordered: $u_i$ is the $i$-th largest child of $u$; similarly for $v$. Then vertex $v$ is the $k$-th largest child of $u$; vertex $u$ is the $l$-th largest child of $v$. Same names for vertices indicate equal vertices.}
	\label{fig:planar:relation_neighboring_vertices}
\end{figure*}
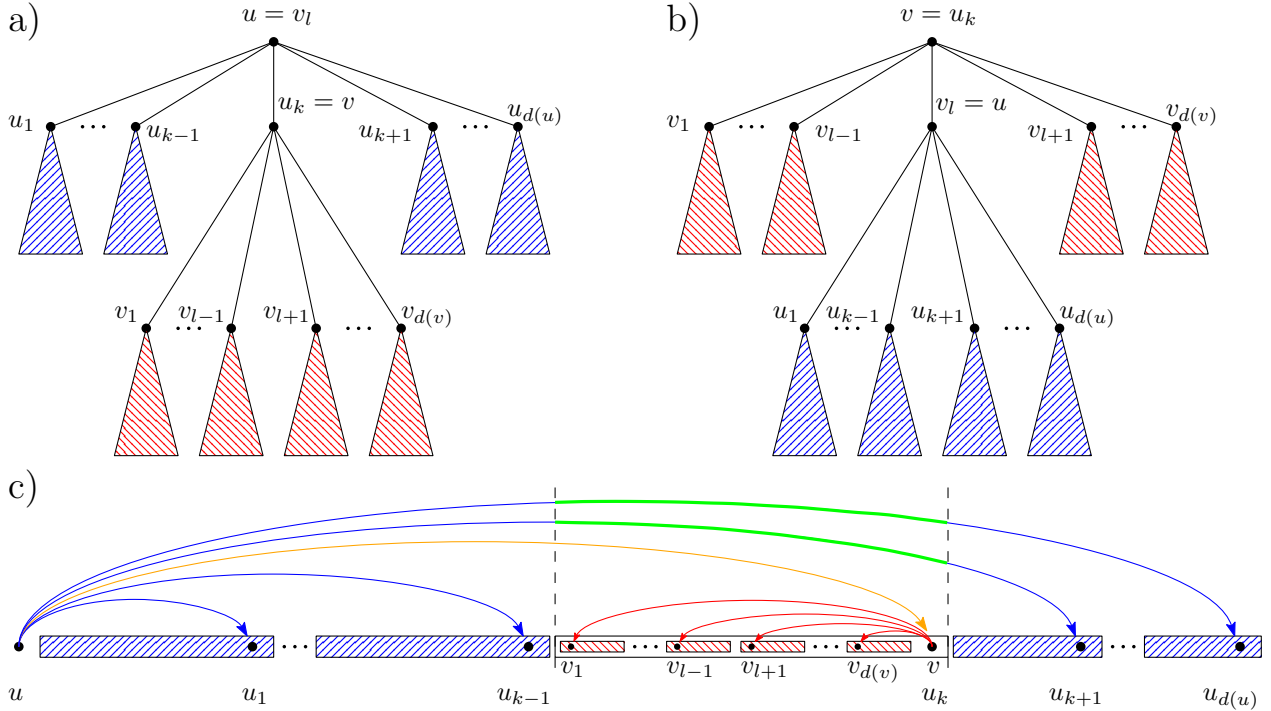


Now we explain how to use Lemma \ref{lem:planar:relation_neighboring_vertices} to solve {\tt planar MaxLA}. In order to calculate efficiently the value $\DMaxProj[\Rtree[v]]$ using the value $\DMaxProj[\Rtree[u]]$, we use a data structure, denoted as $\magadjlist$, similar to an adjacency list. $\magadjlist$ has one entry for each vertex. For any $u\in V$, $\magadjlist[u]$ contains $\degree{u}$-many tuples, $\magadjlist[u][i]$ where $1\le i\le\degree{u}$, each of five elements. Let $v_1,\dots,v_{\degree{u}}$ be the neighbors of $u$ ordered non-increasingly by size. The tuple relative to $(u,v_i)\in E(\Rtree[u])$ is $\magadjlist[u][i]$ and contains (1) $v_i$, (2) the size of $\SubRtree{i}[u]$, (3) the index $i$ which is the position of $v$ in the list of sorted neighbors of $u$, (4) the position of $u$ in the list of sorted neighbors of $v_i$, and (5) the sum of the sizes of the subtrees $\SubRtree{1}[u],\dots,\SubRtree{i-1}[u]$. See the proof of Theorem \suppref{4.1} in the Supplementary Materials for further details. The following theorem outlines the algorithm.

\begin{theorem}
\label{thm:planar:algorithm}
There is an algorithm that, for any tree $\Ftree$, solves {\tt planar MaxLA} in time and space $\bigO{n}$ giving one of the possibly many maximum planar arrangements.
\end{theorem}
\begin{proof}
The algorithm (Algorithm \suppref{4.1}, Supplementary material) first builds structure $\magadjlist$, in order to apply Equation \ref{eq:planar:relation_neighboring_vertices} in constant time. Its construction is done in time and space $\bigO{n}$. The optimum root can be found with a BFS traversal on the tree starting at an arbitrary vertex $w$ with the aid of $\magadjlist$. The value of $\DMaxProj[\Rtree[w]]$ is calculated in $\bigO{n}$ (Corollary \ref{cor:projective:algorithm}). Every step of the traversal has cost $\bigO{1}$ using $\magadjlist$ and Lemma \ref{lem:planar:relation_neighboring_vertices}. Once the optimum root is found, say $z$, we already know the cost, and a maximum planar arrangement is calculated in $\bigO{n}$ as a maximum projective arrangement for $\Rtree[z]$. We detail such algorithm and give its proof of correctness and cost in the Supplementary material (Theorem \suppref{4.1}, Algorithm \suppref{4.1}).
\end{proof}


In the algorithm above, we showed that the optimum root to solve {\tt planar MaxLA} can be found in $\bigO{n}$. Although the algorithm is efficient, it is also difficult to understand since it needs to build $\magadjlist$, traverse the whole tree to find the optimum root, and there is no immediate knowledge on the kind of vertex that it finds. Ideally, the optimum root should be searched using a characterization easier to understand and calculate as it happens in {\tt planar minLA} for trees where the optimal root is a centroidal vertex. We now identify a necessary condition for a vertex to be an optimum root to solve {\tt planar MaxLA}.


Let $\maxroots$ be the set of vertices of $\Ftree$ for which Theorem \ref{thm:planar:quadratic} holds. More formally,
\begin{equation}
\label{eq:planar:vstar}
\maxroots = \{v\in V \;|\; \DMaxPlan=\DMaxProj[\Rtree[v]]\}.
\end{equation}
We now present a characterization of the vertices in $\maxroots$.

\begin{lemma}
\label{lem:planar:max_root_internal_with_leaves}
For any free tree $\Ftree$, a vertex $v\in\maxroots$ is either a leaf or an internal vertex with some leaves adjacent.
\end{lemma}
\begin{proof}
Notice that for $n\le4$ this is trivially true. In the following analysis we assume that $n\ge5$.

Let $w$ be an internal vertex that has no leaf adjacent and let $z$ be the smallest child of $w$. Here we apply Lemma \ref{lem:planar:relation_neighboring_vertices} to show that $\DMaxProj[\Rtree[w]]<\DMaxProj[\Rtree[z]]$. By repeatedly choosing the smallest child starting at $w$ we get a sequence of vertices $u_1$, $\dots$, $u_j$, $u_{j+1}$ in which (1) $u_{i+1}$ is the smallest child of $u_i$, (2) $\DMaxProj[\Rtree[u_i]]<\DMaxProj[\Rtree[u_{i+1}]]$, and (3) $u_j$ is the first to have a leaf adjacent. We can stop at $u_j$ because, by Corollary \ref{cor:projective:characterizing_maximum_arrangements}(\ref{cor:projective:characterizing_maximum_arrangements:ii}), we have that $\DMaxProj[\Rtree[u_j]]=\DMaxProj[\Rtree[u_{j+1}]]$, where $u_{j+1}$ is any leaf of $u_j$.

Now, let $k=\outdegree{z}[w]$. By our assumption that $z$ is not a leaf, $k\ge1$. Let $\arr_1$ be a maximum projective arrangement for $\Rtree[w]$ (Figure \ref{fig:improvable_arrangement}(a)), and let $\arr_2$ be a maximum projective arrangement for $\Rtree[z]$ (Figure \ref{fig:improvable_arrangement}(b)). It is easy to see that $\DMaxProj[\SubRtree{w_i}[w]]$ for all neighbors $w_i$ of $w$ (except $z$) do not change nor does the length of the edges $ww_i$ from $\arr_1$ to $\arr_2$ ($1\le i<\degree{w}$). Then, due to Lemma \ref{lem:planar:relation_neighboring_vertices}, $\DMaxProj[\Rtree[z]] - \DMaxProj[\Rtree[w]] = (k + 1)\Nvert{z}{w} - (n - 1)$ and $\DMaxProj[\Rtree[z]] > \DMaxProj[\Rtree[w]]$ iff,
\begin{equation}
\label{eq:planar:difference_w_to_nonleaf}
k\Nvert{z}{w} > \Nvert{w}{z} - 1.
\end{equation}
If $z$ is not a leaf, Equation \ref{eq:planar:difference_w_to_nonleaf} holds iff $\Nvert{z}{w} > \Nvert{w}{z}$ which trivially holds since $z$ is the smallest child of $w$ and thus subtree $\SubRtree{w}[z]$ is the largest immediate subtree of $\Rtree[z]$. Now, if $z$ was a leaf, that is, if $w$ had an adjacent leaf, it would be easy to see that $\arr_2$ is just the reverse of $\arr_1$. Thus improvement stops at the vertex that has a leaf adjacent to it. More precisely, if $z$ is a leaf Equation \ref{eq:planar:difference_w_to_nonleaf} does not hold since $k=0$ and $\Nvert{w}{z}=1$.
\end{proof}

\begin{figure}
	\centering
\scalebox{0.8}{
\begin{tikzpicture}[ipe stylesheet]
  \pic
     at (142, 696) {ipe disk};
  \node[ipe node]
     at (138, 680) {$w$};
  \draw
    (150, 700) rectangle (230, 692);
  \pic[ipe mark tiny]
     at (234, 696) {ipe disk};
  \pic[ipe mark tiny]
     at (238, 696) {ipe disk};
  \pic[ipe mark tiny]
     at (242, 696) {ipe disk};
  \node[ipe node]
     at (220, 680) {$w_1$};
  \draw
    (246, 700) rectangle (310, 692);
  \node[ipe node]
     at (290, 680) {$w_{j-1}$};
  \draw
    (326, 700) rectangle (402, 692);
  \pic
     at (398, 696) {ipe disk};
  \pic[ipe mark small]
     at (330, 696) {ipe disk};
  \draw
    (328, 698) rectangle (336, 694);
  \pic[ipe mark small]
     at (378, 696) {ipe disk};
  \draw
    (376, 698) rectangle (394, 694);
  \draw
    (141.9999, 696.0001)
     arc[start angle=-135.1971, end angle=-50.2557, x radius=62.2733, y radius=-62.2733];
  \draw
    (142.0004, 695.9996)
     arc[start angle=-135.7073, end angle=-47.0872, x radius=117.422, y radius=-117.422];
  \pic[ipe mark tiny]
     at (358, 696) {ipe disk};
  \pic[ipe mark tiny]
     at (352.0754, 696) {ipe disk};
  \pic[ipe mark tiny]
     at (364.1899, 696) {ipe disk};
  \draw
    (141.9997, 696.0002)
     arc[start angle=-137.2771, end angle=-44.5134, x radius=176.829, y radius=-176.829];
  \pic[ipe mark small]
     at (306, 696) {ipe disk};
  \pic[ipe mark small]
     at (226, 696) {ipe disk};
  \draw
    (330, 700)
     arc[start angle=-135, end angle=-45, x radius=48.0832, y radius=-48.0832];
  \draw
    (378.0003, 699.9998)
     arc[start angle=-142.4289, end angle=-37.5711, x radius=12.6164, y radius=-12.6164];
  \node[ipe node]
     at (118, 696) {$\arr_1$};
  \node[ipe node]
     at (138, 590) {$z$};
  \pic
     at (142, 606) {ipe disk};
  \node[ipe node]
     at (118, 606) {$\arr_2$};
  \node[ipe node, font=\Large]
     at (118, 740) {a)};
  \node[ipe node, font=\Large]
     at (118, 654) {b)};
  \node[ipe node]
     at (236, 680) {$\ge\cdots\ge$};
  \node[ipe node]
     at (314, 680) {$\ge$};
  \draw
    (156, 608) rectangle (210, 604);
  \pic[ipe mark small]
     at (158, 606) {ipe disk};
  \node[ipe node]
     at (154, 590) {$w_{j-1}$};
  \pic[ipe mark tiny]
     at (214, 606) {ipe disk};
  \pic[ipe mark tiny]
     at (220, 606) {ipe disk};
  \pic[ipe mark tiny]
     at (226, 606) {ipe disk};
  \draw
    (232, 608) rectangle (310, 604);
  \pic[ipe mark small]
     at (234, 606) {ipe disk};
  \node[ipe node]
     at (230, 590) {$w_1$};
  \node[ipe node]
     at (190, 590) {$\le\cdots\le$};
  \draw
    (154, 610) rectangle (320, 602);
  \pic
     at (316, 606) {ipe disk};
  \node[ipe node]
     at (312, 590) {$w$};
  \pic[ipe mark small]
     at (342, 606) {ipe disk};
  \draw
    (326, 610) rectangle (346, 602);
  \pic[ipe mark tiny]
     at (364, 604) {ipe disk};
  \pic[ipe mark tiny]
     at (358.075, 604) {ipe disk};
  \pic[ipe mark tiny]
     at (370.19, 604) {ipe disk};
  \node[ipe node]
     at (390, 590) {$z_k$};
  \node[ipe node]
     at (338, 590) {$z_1$};
  \node[ipe node]
     at (350, 590) {$\ge\cdots\ge$};
  \pic[ipe mark small]
     at (398, 606) {ipe disk};
  \draw
    (390, 610) rectangle (402, 602);
  \draw
    (234.0003, 609.9996)
     arc[start angle=-129.7444, end angle=-50.2556, x radius=64.1258, y radius=-64.1258];
  \draw
    (158, 609.9999)
     arc[start angle=-133.409, end angle=-46.591, x radius=114.959, y radius=-114.959];
  \draw
    (142.0001, 606)
     arc[start angle=-142.296, end angle=-40.338, x radius=112.011, y radius=-112.011];
  \draw
    (141.9994, 606.0005)
     arc[start angle=-142.2955, end angle=-39.9956, x radius=128.43, y radius=-128.43];
  \draw
    (142.0006, 605.9995)
     arc[start angle=-142.2961, end angle=-39.4944, x radius=163.801, y radius=-163.801];
  \node[ipe node]
     at (338, 680) {$\le\cdots\le$};
  \node[ipe node]
     at (326, 680) {$z_k$};
  \node[ipe node]
     at (378, 680) {$z_1$};
  \node[ipe node]
     at (396, 680) {$z$};
  \node[ipe node]
     at (326, 590) {$\ge$};
\end{tikzpicture}
}
	\caption{Proof of Lemma \ref{lem:planar:max_root_internal_with_leaves}. Vertex $w$ has neighbors $w_1$, $\dots$, $w_j$; no neighbor of $w$ is a leaf. Vertex $z=w_j$ has children $z_1$, $\dots$, $z_k$ in $\Rtree[w]$. The comparison symbols $\ge$ and $\le$ in the diagram indicate order of size of the subtrees rooted at the vertices at every side of the comparison symbols. a) A maximum projective arrangement $\arr_1$ of $\Rtree[w]$. b) A maximum projective arrangement $\arr_2$ of $\Rtree[z]$.}
	\label{fig:improvable_arrangement}
\end{figure}
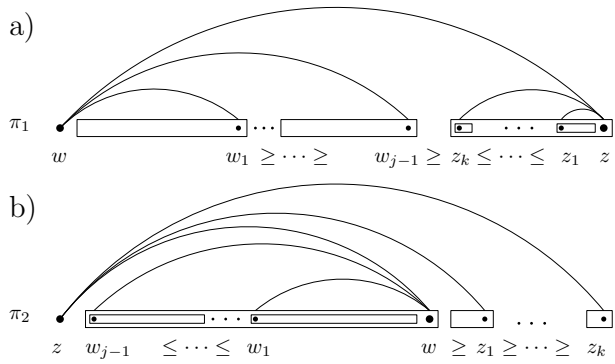

The proof of Lemma \ref{lem:planar:max_root_internal_with_leaves} does not guarantee that repeatedly choosing the smallest child yields $\DMaxPlan$. It only guarantees improvement over the starting vertex when said vertex has no leaves attached. Furthermore, from Lemma \ref{lem:planar:max_root_internal_with_leaves}, $2\le|\maxroots|$ (when $n \geq 2$); moreover, in the worst case, that is, in a star tree, $|\maxroots|=n$. Since the vertices in $\maxroots$ are either leaves or their only neighbor, we can construct, in time $\bigO{n}$, the set $\maxrootsk\subseteq V(\Ftree)$ of vertices that have a leaf attached. It is easy to see that $1\le|\maxrootsk|\le \min\{n/2,\numleaves\}$, where $\numleaves$ is the number of leaves of $\Ftree$. In order to solve {\tt planar MaxLA}, we can find the maximum root by checking all vertices in $\maxrootsk$ in time $\bigO{|\maxrootsk|n}$.

The next corollary highlights a significant difference between the structure of maximum projective arrangements and the structure of maximum planar arrangements. Although the strategies to solve these related problems have steps in common, there are differences in the solution. While a maximum projective arrangement can only have leaves at one end if the root vertex has a leaf adjacent (Corollary \ref{cor:projective:characterizing_maximum_arrangements}), a maximum planar arrangement always has leaves at one of the ends.
\begin{corollary}
\label{cor:planar:root_and_leaves_in_end_position}
In a maximum planar arrangement $\arr$ of $\Ftree$,
\begin{enumerate}[(i)]
\item An internal vertex $v$ with leaves is placed at one end of the linear arrangement,
\item All leaves of $v$ are arranged consecutively at the other end of the linear arrangement.
\end{enumerate}
\end{corollary}
\begin{proof}
Theorem \ref{thm:planar:quadratic} indicates that a maximum planar arrangement is a maximum projective arrangement of some rooted tree. Then this corollary follows from combining Corollary \ref{cor:projective:characterizing_maximum_arrangements} with Lemma \ref{lem:planar:max_root_internal_with_leaves}.
\end{proof}

Notice that maximum planar arrangements also satisfy properties in Corollary \ref{cor:projective:characterizing_maximum_arrangements}(\ref{cor:projective:characterizing_maximum_arrangements:iii},\ref{cor:projective:characterizing_maximum_arrangements:iv}).

Although solving {\tt planar MaxLA} for general trees requires traversing the whole tree to identify the optimum root, it can be solved quite easily in caterpillar trees. The next corollary explains how to construct a maximum planar arrangement for any caterpillar tree, using the fact that all caterpillars are {\em graceful trees} proven by Rosa \cite{Rosa1967a}. A tree is said to be {\em graceful} if one can find a (bijective) labeling $\phi\;:\; V \rightarrow \{0,\dots,n-1\}$ such that each edge $uv$ is uniquely identified by the value $|\phi(u) - \phi(v)|$; see \cite{Gallian2018a} for further details. Then, we define the arrangement $\arr$ as $\arr(u)=\phi(u)+1$ and then a graceful tree is one that can be arranged such that there is exactly one edge of length $1$, one of length $2$, and so on, until one edge of length $n-1$. It is easy to see, then, that an arrangement of a tree defined from a graceful labeling has cost exactly ${n \choose 2}$.

The next corollary also generalizes a previous result \cite{Ferrer2013a}, namely
\begin{equation}
\label{eq:introduction:upper_bound_maximum_planar}
\begin{split}
\max_{\Ftree \in \FreeTrees} \{\DMaxPlan\}
	&= \DMaxPlan[\Star] \\
	&= \DMaxPlan[\Linear] = {n \choose 2}.
\end{split}
\end{equation}

\begin{corollary}
\label{cor:planar:formula_caterpillars}
~\\
\begin{enumerate}[(i)]
\item \label{cor:planar_formula_caterpillars:1} Rosa's construction \cite{Rosa1967a} to show that caterpillar trees are graceful corresponds to maximum planar arrangements of the same trees,

\item \label{cor:planar_formula_caterpillars:2} For any caterpillar tree $\Ftree$, $\DMaxPlan = {n \choose 2}$,

\item \label{cor:planar_formula_caterpillars:3} The maximum $\DMaxPlan$ over all trees is achieved at least by caterpillar trees,

\item \label{cor:planar_formula_caterpillars:4} The endpoints of the caterpillar's backbone are in $\maxroots$.
\end{enumerate}
\end{corollary}
\begin{proof}
~\\
\begin{enumerate}[(i)]
\item First, Rosa \cite{Rosa1967a} showed that all caterpillar trees $\Ftree$ are graceful with the following construction. Identify the vertices of the caterpillar's backbone; let $v$ be the vertex at one end of the backbone. Place $v$ at one end of the arrangement, and its leaves at the other end. Let $w$ be the only neighbor of $v$ that has not yet been placed. Place $w$ next to the leaves of $v$, and arrange the leaves of $w$ next to $v$. Continue until there are no more vertices to arrange. It is easy to see that Rosa's construction for caterpillar trees is planar and follows the characterization of a maximum projective arrangement of $\Rtree[v]$ (Theorem \ref{thm:projective:shape_maximum}, Figure \ref{fig:projective:shape_maximum}).

\item Follows from (\ref{cor:planar_formula_caterpillars:1}) and the definition of graceful trees.

\item By (\ref{cor:planar_formula_caterpillars:2}) and Equation \ref{eq:introduction:upper_bound_maximum_planar}, it immediately follows that all caterpillar trees maximize the value of {\tt MaxLA} over all free trees.

\item Finally, it trivially follows from (\ref{cor:planar_formula_caterpillars:1}) that the endpoints of the backbone are in $\maxroots$ because the ends of the backbone have the role of a root vertex in the characterization of a maximum planar arrangement (Theorem \ref{thm:planar:quadratic}, Equation \ref{eq:planar:vstar}).
\end{enumerate}
\end{proof}
\section{Conclusions}
\label{sec:conclusions}

In this paper we have tackled {\tt projective MaxLA} and {\tt planar MaxLA}. We have presented an algorithm to solve {\tt projective MaxLA} in time and space $\bigO{n}$ which, besides, follows an strategy almost identical to that used to solve {\tt projective minLA}. However, the approach we used to solve {\tt planar MaxLA} differs from the approach used to solve {\tt planar minLA}: we were unable to characterize the vertices in $\maxrootsk$ and thus, by extension, the vertex $v$ in the outline above (Section \ref{sec:outline}) for most trees, the exception being caterpillar trees. Moreover, the bounds to the size of $\maxrootsk$ are not sufficiently low to simplify the algorithm presented here. Future work should study stronger characterizations of $\maxrootsk$ and search for broader classes of trees for which $|\maxrootsk|=\bigO{1}$.

% ...

\section*{Acknowledgments}
We thank M. Mora for helpful comments and discussions. We thank an anonymous reviewer for pointing us to \cite{Goldman1971a} and to the use of Goldman's algorithm. We owe them the clever proposal of the general algorithm in Section \ref{sec:outline}. The authors are supported by a recognition 2021SGR-Cat (LQMC) from AGAUR (Generalitat de Catalunya). LAP is funded by Secretaria d'Universitats i Recerca de la Generalitat de Catalunya and the Social European Fund. JLE is funded by the grant PID2019-109137GB-C22 from MINECO. 

\bibliographystyle{plain}

\end{document}


\allowdisplaybreaks

\title{Supplementary Material -- The Maximum Linear Arrangement Problem for trees under projectivity and planarity}

\author[1]{Llu\'is Alemany-Puig\thanks{lluis.alemany.puig@upc.edu}}
\author[1]{Juan Luis Esteban\thanks{esteban@cs.upc.edu}}
\author[1]{Ramon Ferrer-i-Cancho\thanks{rferrer@cs.upc.edu}}
\affil[1]{Quantitative, Mathematical and Computational Linguistics Research Group, Computer Science Department, Universitat Polit\`ecnica de Catalunya}

\date{}
\maketitle

\section{Introduction}
\label{sec:introduction}

Here we give detailed proofs of some of the Corollaries and Theorems from the main text and pseudocode for the algorithms to solve the problems tackled in this paper. The theorems, corollaries and lemmas from the main text are references here in \maintextref{red, bold-face and emphasized text}. We refer the reader to the main text for the meaning of notation.

\section{Sorting of subtrees by size}
\label{sec:sorting}

The algorithms presented here require sorting the subtrees of a rooted tree $\Rtree$ non-increasingly by size. To this aim, we use a data structure we denote by $\subtreesizelist^\Root$. It has $n$ entries: for any vertex $u\in V$, $\subtreesizelist^\Root[u]$ contains $\degree{u}-1$ tuples $(v,\Nvert{\Root}{v})$ where the $i$-th tuple is the root of $i$-th largest immediate subtree of $\SubRtree{u}$ and its size $\Nvert{\Root}{u}$. The superscript $\Root$ in $\subtreesizelist^\Root$ indicates that the list is {\em rooted} on $\Root$, i.e., there are no tuples of the form $(p, \Nvert{\Root}{p})$ where $p=\parent{u}$ denotes the parent of $u$ in $\Rtree$\footnote{The parent of $u\neq\Root$ in $\Rtree$ is the only vertex $v$ such that $(v,u)\in E(\Rtree)$.}. $\subtreesizelist^\Root$ can be constructed in time $\bigO{n}$ and space $\bigO{n}$ using a result by Hochberg and Stallmann \cite{Hochberg2003a}, which was further studied in \cite{Alemany2022a}, and is stated next.

\begin{proposition}[Hochberg and Stallmann \cite{Hochberg2003a}]
\label{prop:preliminaries:calculate_suvs}
Let $\Ftree=(V,E)$ be a free tree. The values in the set
\begin{equation*}
\{ (u, v, \Nvert{u}{v}), (v, u, \Nvert{v}{u}) \;|\; uv\in E \},
\end{equation*}
where $(x,y, \Nvert{x}{y})$ denotes a 3-tuple in which $xy\in E$, are $\bigO{n}$-time and $\bigO{n}$-space computable.
\end{proposition}

Algorithms 2.1 and 4.2 in \cite{Alemany2022a} show pseudocode for Proposition \ref{prop:preliminaries:calculate_suvs}. The former (\cite[Algorithm 2.1]{Alemany2022a}) is used by Algorithm \ref{algo:planar:maximum:calculate_M} under the name \textsc{compute\_s\_ft}. The latter (\cite[Algorithm 4.2]{Alemany2022a}) is used by Algorithm \ref{algo:projective:maximum} to construct $\subtreesizelist^\Root$. In \cite{Alemany2022a}, $\subtreesizelist^\Root$ was used to calculate $\DminProj$; here we need it to compute $\DMaxProj$. Readers will find pseudocode to construct $\subtreesizelist^\Root$ in \cite[Algorithm 4.2]{Alemany2022a}. That algorithm is used here under the name \textsc{sorted\_adjacency\_list\_rt}.

\section{{\tt Projective MaxLA}}

Algorithm \ref{algo:projective:maximum} solves {\tt projective MaxLA}. It is a modification of algorithms previously used to solve {\tt projective minLA} \cite[Algorithm 4.1]{Alemany2022a}. We also justify its correctness and cost more formally.

\begin{theorem}[\maintextref{Corollary 4.4}]
For any rooted tree $\Rtree$, Algorithm \ref{algo:projective:maximum} solves {\tt projective MaxLA} in time $\bigO{n}$, space $\bigO{n}$.
\end{theorem}
\begin{proof}
Algorithm \ref{algo:projective:maximum} is a modification of an algorithm previously used to solve {\tt projective minLA} \cite[Algorithm 4.1]{Alemany2022a}. The algorithm to solve {\tt projective MaxLA} is based on \maintextref{Theorem 4.1}. The algorithm takes $\SubRtree{u}$ and uses the subtree-size adjacency list $\subtreesizelist^\Root[u]$ to arrange its immediate subtrees in a non-increasing fashion in the interval $[a,b]$ passed as parameter (\maintextref{Theorem 4.1(i)}). In the first call, the interval is $[1,n]$. Since we have to construct a left (resp. right) branching arrangement (\maintextref{Theorem 4.1}), we place the root of the subtree at either position $a$ or $b$; the exact end is decided using the parameter $\tau$ which indicates the side in which $u$ has been placed with respect to its parent in the arrangement, either {\tt left} or {\tt right} (\maintextref{Theorem 4.1(ii)}). The interval of positions within the arrangement of the $i$-th subtree when $\tau=\mathtt{left}$ is
\begin{equation*}
\left[
a + \sum_{j=1}^{i-1} \largestNvert{\Root}{u}{j} + 1, a + \sum_{j=1}^i \largestNvert{\Root}{u}{j}
\right]
\end{equation*}
and when $\tau=\mathtt{right}$ the interval is
\begin{equation*}
\left[
b - \sum_{j=1}^i \largestNvert{\Root}{u}{j}, b - \sum_{j=1}^{i-1} \largestNvert{\Root}{u}{j} - 1
\right].
\end{equation*}

It is easy to see that its time complexity is $\bigO{n}$. Its space complexity is also $\bigO{n}$ given the need to construct $\subtreesizelist^{\Root}$.
\end{proof}

\begin{algorithm*}
	\small
	\caption{{\tt Projective MaxLA}.}
	\label{algo:projective:maximum}
	\DontPrintSemicolon
	
	\SetKwProg{Fn}{Function}{ is}{end}
	\Fn{\textsc{ArrangeMaximumProjective}$(\Rtree)$} {
		\KwIn{$\Rtree$ rooted tree at $\Root$.}
		\KwOut{A maximum projective arrangement $\arr$.}
		
		$\subtreesizelist^\Root \gets \textsc{sorted\_adjacency\_list\_rt}(\Rtree)$ \label{algo:projective:maximum:root_list} \tcp{\cite[Algorithm 4.2]{Alemany2022a}}
		$\arr \gets \{0\}^n$ \tcp{empty arrangement}
		\textsc{Arrange\_Rec}$(\subtreesizelist^\Root, \Root, \mathtt{right}, 1, n, \arr)$ \label{algo:projective:maximum:first_rec_call} \tcp{The starting side is arbitrary.}
		\Return $\arr$
	}
	
	\SetKwProg{Fn}{Function}{ is}{end}
	\Fn{\textsc{Arrange\_Rec}$(\subtreesizelist^\Root, u, \tau, a,b, \arr)$} {
		\KwIn{$\subtreesizelist^\Root$ as described in Section \ref{sec:sorting}; $u$ the root of the subtree to be arranged; $\tau$ the side (relative position of $u$ with respect to its parent); $[a,b]$ interval of positions of $\arr$ where to embed $\SubRtree{u}$; $\arr$ a partially-constructed arrangement.}
		\KwOut{$\arr$ updated with the optimal projective arrangement for $\SubRtree{u}$ in $[a,b]$.}
		
		$C_u \gets \subtreesizelist^\Root[u]$ \tcp{the children of $u$ sorted non-increasingly by size}
		$S\gets 0$ \tcp{Cumulative size of the subtrees}
		
		\lIf {$\tau = \mathtt{left}$} {
			$\tau_{\mathrm{next}} \gets$ {\tt right}
		}
		\lElse {
			$\tau_{\mathrm{next}} \gets$ {\tt left}
		}
		
		\For {$i$ from $1$ to $|C_u|$} { \label{algo:projective:maximum:for_loop}
			$v, n_v \gets C_u[i]$ \tcp{the $i$-th child of $u$, and its size $n_v=\Nvert{\Root}{v}$}
			
			\lIf {$\tau$ is $\mathtt{left}$} {
				$a_{\mathrm{next}}\gets a + S + 1$;
				$b_{\mathrm{next}}\gets a_{\mathrm{next}} + n_v - 1$
			}
			\lElse {
				$b_{\mathrm{next}}\gets b - S - 1$;
				$a_{\mathrm{next}}\gets b_{\mathrm{next}} - n_v + 1$
			}
			
			\textsc{Arrange\_Rec}$(\subtreesizelist^\Root, v, \tau_{\mathrm{next}}, a_{\mathrm{next}}, b_{\mathrm{next}}, \arr)$
			
			$S \gets S + n_v$
		}
		\lIf {$\tau = \mathtt{left}$} {
			$\arr(u) \gets a$
		}
		\lElse {
			$\arr(u) \gets b$
		}
	}
\end{algorithm*}

\section{{\tt Planar MaxLA}}

Here we present Algorithm \ref{algo:planar:maximum} to solve {\tt Planar MaxLA}. We also prove its correctness and cost more formally.

\begin{theorem}[\maintextref{Theorem 5.3}]
For any free tree $\Ftree$, Algorithm \ref{algo:planar:maximum} solves {\tt planar MaxLA} in time and space $\bigO{n}$.
\end{theorem}
\begin{proof}
The main idea behind Algorithm \ref{algo:planar:maximum} is that it performs a BFS traversal on the tree and, for every traversed edge $uv$, say, from $u$ to $v$, apply \maintextref{Lemma 5.2} so as to calculate the value $\DMaxProj[\Rtree[v]]$ in time $\bigO{1}$ (using the fact that $\DMaxProj[\Rtree[u]]$ is known). The BFS starts at an arbitrary vertex, say $w$, for which we calculate $\DMaxProj[\Rtree[w]]$ in time $\bigO{n}$ using Algorithm \ref{algo:projective:maximum}. Once the traversal has finished the algorithm will have calculated all values $\DMaxProj[\Rtree[x]]$ for all $x\in V$.

Algorithm \ref{algo:planar:maximum} simply calls Algorithm \ref{algo:planar:optimal_root} which first constructs, using Algorithm \ref{algo:planar:maximum:calculate_M}, an adjacency list similar to $\subtreesizelist^\Root$, denoted as $\magadjlist$ (described in \maintextref{Section 5}). After that, Algorithm \ref{algo:planar:optimal_root} performs the BFS traversal explained above to calculate in constant time $\DMaxProj[\Rtree[v]]$ from $\DMaxProj[\Rtree[u]]$ using the values stored in $\magadjlist$ and applying them with \maintextref{Lemma 5.2}.

The structure $\magadjlist$ contains a tuple for every directed edge. Given a directed edge $(u,v)$, the tuple is of the form $(v,\Nvert{u}{v},\indexof{u}{v},\indexof{v}{u},\sumlargest{u}{v})$, where $v$ is the root of the $i$-th largest immediate subtree of $\Rtree[u]$ and $\sumlargest{u}{v}$ denotes the sum of the sizes of the first $\indexof{u}{v}$ largest immediate subtrees of $\Rtree[u]$,
\begin{equation*}
\sumlargest{u}{v} = \sum_{j=1}^{\indexof{u}{v}} \singlelargestNvert{u}{j}.
\end{equation*}
Notice that for every edge $uv$, the entries
\begin{align*}
\magadjlist[u][\indexof{u}{v}]
	&= (v, \Nvert{u}{v}, \indexof{u}{v}, \indexof{v}{u}, \sumlargest{u}{v}) \\
\magadjlist[v][\indexof{v}{u}]
	&= (u, \Nvert{v}{u}, \indexof{v}{u}, \indexof{u}{v}, \sumlargest{v}{u})
\end{align*}
are related since the third value in $\magadjlist[u][\indexof{u}{v}]$ is the same as the fourth value in $\magadjlist[v][\indexof{v}{u}]$ and the third in $\magadjlist[v][\indexof{v}{u}]$ is the same as the fourth in $\magadjlist[u][\indexof{u}{v}]$.

It can be seen in Algorithm \ref{algo:planar:maximum:calculate_M} that the calculation of the values $\Nvert{u}{v}$, $\indexof{u}{v}$, $\sumlargest{u}{v}$ for the tuples in entry $\magadjlist[u]$ is correct. It remains to ascertain that the value $\indexof{v}{u}$ is also calculated correctly in time $\bigO{n}$. These values cannot be calculated directly in the first for loop. This is why they are added as `0' in line \ref{algo:planar:maximum:calculate_M:update_M:first}, and thus an auxiliary list $J$ (line \ref{algo:planar:maximum:calculate_M:J}) is used instead: it is first filled in the first loop (line \ref{algo:planar:maximum:calculate_M:update_J}), and is used in the second loop starting at line \ref{algo:planar:maximum:calculate_M:for_loop:2} to add the values $\indexof{v}{u}$ in the corresponding tuples in $\magadjlist[u]$ (line \ref{algo:planar:maximum:calculate_M:update_M:last}). At the end of the first for loop, the auxiliary list $J$ contains the same tuples as $S$ except that $J$ contains extra indices $\indexof{v}{u}$: notice that values $(v,u,n-\Nvert{u}{v},\indexof{u}{v})$ are the same as $(v,u,\Nvert{v}{u},\indexof{u}{v})$, and, furthermore, the same as $(u,v,\Nvert{u}{v},\indexof{v}{u})$. Now, let $J_u\subseteq J$ be the subset of tuples from $J$ starting with vertex $u$. The relative order of the tuples from $J_u$ (in $J$), is the same from the corresponding tuples in $S$ with the exception of ties among neighbors $v$ of $u$ with the same size $\Nvert{u}{v}$. The order in which they appear in $J$ is irrelevant, that is, it is easy to see that for a given $u\in V$ and a subset of neighbors $\{v_1,\dots,v_k\}$ of $u$ such that $\Nvert{u}{v_i}=\Nvert{u}{v_j}$ for all $i,j\in[1,k]$ we have that $\indexof{v_i}{u}=1$. It remains to distribute them appropriately among the tuples in $\magadjlist[u]$. This is done in the second for loop (line \ref{algo:planar:maximum:calculate_M:for_loop:2}).

The size complexity of $\magadjlist$ is clearly $\bigO{n}$. The extra space complexity needed to store the values of $\DMaxProj[\Rtree[x]]$ for all $x\in V$ is also $\bigO{n}$. The call to Algorithm \ref{algo:projective:maximum}, the BFS traversal, and the application of \maintextref{Lemma 5.2} all require time complexity $\bigO{n}$.
\end{proof}

\begin{algorithm*}
	\small
	\caption{{\tt Planar MaxLA}.}
	\label{algo:planar:maximum}
	\DontPrintSemicolon
	
	\KwIn{$\Ftree$ free tree.}
	\KwOut{A maximum planar arrangement $\arr$.}
	
	\SetKwProg{Fn}{Function}{ is}{end}
	\Fn{\textsc{ComputeMaximumPlanarD}$(\Ftree)$} {
		$s \gets \textsc{ComputeOptimalRoot}(\Ftree)$ \tcp{Algorithm \ref{algo:planar:optimal_root}.}
		\Return \textsc{ArrangeMaximumProjective}$(\Rtree[s])$ \tcp{Algorithm \ref{algo:projective:maximum}.}
	}
\end{algorithm*}

\begin{algorithm*}
	\small
	\caption{Finding an optimal root.}
	\label{algo:planar:optimal_root}
	\DontPrintSemicolon
	
	\KwIn{$\Ftree$ free tree.}
	\KwOut{A vertex $\Root$ that maximizes $\DMaxPlan[\Rtree]$.}
	
	\SetKwProg{Fn}{Function}{ is}{end}
	\Fn{\textsc{ComputeOptimalRoot}$(\Ftree)$} {
		
		$\magadjlist \gets \textsc{Compute\_}\magadjlist(\Ftree)$ \label{line:algo:planar:maximum:call_M} \tcp{Algorithm \ref{algo:planar:maximum:calculate_M}}
		
		$w \gets$ choose an arbitrary vertex \label{line:algo:planar:maximum:starting_vertex} \;
		$Q\gets\{w\}$ \tcp{the queue of the breadth-first search}
		$vis\gets\{0\}^n$; $vis[w]\gets 1$ \tcp{vertices marked as not visited, except $w$}
		
		$D \gets \{0\}^n$ \tcp{an $n$-zero vector}
		$D[w] \gets \DMaxProj[\Rtree[w]]$ \tcp{calculated using Algorithm \ref{algo:projective:maximum}}
		
		\While {$|Q|>0$} {
			$u\gets front(Q)$ \tcp{the next vertex in the queue}
			$Q = Q\setminus\{u\}$ \tcp{remove $u$ from the queue}
			\For {$X\in\magadjlist[u]$} {
				\If {$vis[X.v] \neq 1$ and $\degree{X.v} > 1$} { \tcp{excluding leaves because they do not belong to $\maxrootsk$}
				
					$D[X.v] \gets \DMaxProj[\Rtree[X.v]]$ \tcp{\maintextref{Lemma 5.2}. Use $\magadjlist[X.v][X.\indexof{v}{u}].\sumlargest{v}{u}$ and $D[u]$}
					
					$vis[X.v] \gets 1$ \tcp{Mark vertex $X.v$ as visited}
					$Q \gets Q \cup \{X.v\}$ \tcp{Append $X.v$ to the queue}
				}
			}
		}
		$s \gets \argmax_{u\in V}\{D[u]\}$ \tcp{$D[u]$ contains $\DMaxProj[\Rtree[u]]$ for all $u\in V$}
		\Return $s$
	}
\end{algorithm*}

\begin{algorithm*}
	\small
	\caption{Calculation of $\magadjlist$.}
	\label{algo:planar:maximum:calculate_M}
	\DontPrintSemicolon
	
	\KwIn{$\Ftree$ free tree.}
	\KwOut{$\magadjlist$, where $\magadjlist[u]$ contains $\degree{u}$-many tuples of the form $(v,\Nvert{u}{v}, \indexof{u}{v}, \indexof{v}{u}, \sumlargest{u}{v})$.}
	
	\SetKwProg{Fn}{Function}{ is}{end}
	\Fn{\textsc{Compute\_}$\magadjlist(\Ftree)$} {
		
		$S\gets \textsc{compute\_s\_ft}(\Ftree)$ \tcp{\cite[Algorithm 2.1]{Alemany2022a}, Proposition \ref{prop:preliminaries:calculate_suvs}}
		Sort the tuples $(u,v, \Nvert{u}{v})$ in $S$ non-increasingly by $\Nvert{u}{v}$ using counting sort \cite{Cormen2001a} \label{algo:sorting_suvs}\;
		
		$J\gets \emptyset$ \tcp{$J$ is used to compute the indices $\indexof{v}{u}$ in every entry of $\magadjlist[u]$} \label{algo:planar:maximum:calculate_M:J}
		
		$\magadjlist\gets\{\emptyset\}^n$ \tcp{Firstly, fill $\magadjlist$ partially}
		\For {$(u,v, \Nvert{u}{v}) \in S$} { \label{algo:planar:maximum:calculate_M:for_loop:1}
			$k_u \gets |\magadjlist[u]|$ \tcp{The current size of $\magadjlist[u]$}
			$g_u \gets \magadjlist[u][k_u - 1].\sumlargest{u}{v}$ \tcp{The value $\sumlargest{u}{v}$ in the last tuple of $\magadjlist[u]$}
			\tcp{By construction, $k_u =\indexof{u}{v}$}
			$\magadjlist[u] \gets \magadjlist[u] \cup (v, \Nvert{u}{v}, k_u, 0,  \Nvert{u}{v} + g_u)$ \tcp{Append at end in $\bigO{1}$.} \label{algo:planar:maximum:calculate_M:update_M:first}
			
			\tcp{Recall that $\Nvert{v}{u} = n - \Nvert{u}{v}$}
			
			$J \gets J \cup (v,u, n - \Nvert{u}{v}, k_u )$ \tcp{Append at end in $\bigO{1}$.} \label{algo:planar:maximum:calculate_M:update_J}
		}
		
		Sort the tuples $(u,v, \Nvert{u}{v}, i)$ in $J$ (where $i=\indexof{v}{u}$) non-increasingly by $\Nvert{u}{v}$ with counting sort \cite{Cormen2001a} \;
		
		\tcp{Fill the missing indices $\indexof{v}{u}$ in the tuples of $\magadjlist[u]$ using $J$}
		$I \gets \{1\}^n$ \;
		\For {$ (u,v, \Nvert{u}{v}, i) \in J $} { \label{algo:planar:maximum:calculate_M:for_loop:2}
			$\magadjlist[u][ I[u] ].\indexof{v}{u} \gets i$ \label{algo:planar:maximum:calculate_M:update_M:last}\;
			$I[u] \gets I[u] + 1$
		}
		
		\Return $\magadjlist$
	}
\end{algorithm*}

\bibliographystyle{plain}